\documentclass[a4paper,UKenglish,cleveref, autoref, thm-restate]{lipics-v2021}

\usepackage{microtype}
\usepackage[T1]{fontenc}
\usepackage{amsmath}
\usepackage{amsfonts}
\usepackage{amsthm}
\usepackage{mathtools}
\usepackage{thmtools}
\usepackage{thm-restate}
\usepackage{xcolor}
\usepackage{dsfont}
\usepackage{tikz}
\usepackage{xspace}
\usepackage{enumitem}
\usepackage[noend,noline,ruled,linesnumbered]{algorithm2e}

\usepackage[utf8]{inputenc}
\usepackage[T1]{fontenc}
\usepackage[capitalise]{cleveref}

\usepackage{forest}

\graphicspath{ {figures/} }

\newcommand{\para}[1]{\subparagraph*{#1}}

\newcommand{\modulo}{\operatorname{mod}}
\newcommand{\ceil}[1]{\left\lceil{#1}\right\rceil}
\newcommand{\floor}[1]{\left\lfloor{#1}\right\rfloor}

\newcommand{\cc}{\textsc{Congested Clique }}
\newcommand{\eps}{\varepsilon}
\renewcommand{\S}{\mathcal{S}}
\newcommand{\B}{\mathcal{B}}

\newcommand{\T}{\mathcal{C}}

\newcommand{\K}{\mathcal{K}}

\newcommand{\SA}{\mathsf{SA}}
\newcommand{\LCP}{\mathsf{LCP}}
\newcommand{\RMQ}{\mathsf{RMQ}}

\newcommand{\rank}{\mathsf{rank}}
\newcommand{\PM}{\mathsf{PM}}

\newcommand{\N}{\mathbb{N}}

\newtheorem{problem}[definition]{Problem}
\newtheorem{fact}[definition]{Fact}

\crefname{problem}{Problem}{Problems}
\crefname{algorithm}{Algorithm}{Algorithms}

\newcommand{\sub}{\subseteq}

\definecolor{niceblue}{rgb}{.392,.584,.929}

\author{Shay Golan}{Reichman University and University of Haifa, Israel \and \url{https://sites.google.com/view/shaygolan}}{golansh1@biu.ac.il}{https://orcid.org/0000-0001-8357-2802}{supported by Israel Science Foundation grant 810/21.}

\author{Matan Kraus}{Bar Ilan Univesity, Israel}{matan3@gmail.com}{https://orcid.org/0000-0002-2989-1113}{supported by the ISF grant no. 1926/19, by the BSF grant 2018364, and by the ERC grant MPM under the EU's Horizon 2020 Research and Innovation Programme (grant no. 683064).}

\authorrunning{S. Golan and M. Kraus}

\Copyright{Shay Golan and Matan Kraus}



\nolinenumbers

\EventEditors{John Q. Open and Joan R. Access}
\EventNoEds{2}
\EventLongTitle{42nd Conference on Very Important Topics (CVIT 2016)}
\EventShortTitle{CVIT 2016}
\EventAcronym{CVIT}
\EventYear{2016}
\EventDate{December 24--27, 2016}
\EventLocation{Little Whinging, United Kingdom}
\EventLogo{}
\SeriesVolume{42}
\ArticleNo{23}

\begin{document}

    \title{String Problems in the Congested Clique Model}

\maketitle
	\begin{abstract}
    In this paper we present algorithms for several string problems in the \cc model.
    In the \cc model, $n$ nodes (computers) are used to solve some problem.
    The input to the problem is distributed among the nodes, and the communication between the nodes is conducted in \emph{rounds}.
    In each round, every node is allowed to send an $O(\log n)$-bit message to every other node in the network.

    We consider three fundamental string problems in the \cc model.
    First, we present an $O(1)$ rounds algorithm for string sorting that supports strings of arbitrary length.
    Second, we present an $O(1)$ rounds combinatorial pattern matching algorithm.
    Finally, we present an $O(\log\log n)$ rounds algorithm for the computation of the suffix array and the corresponding Longest Common Prefix array of a given string.

	\end{abstract}

\clearpage
\setcounter{page}{1}

	\section{Introduction }\label{sec:intro}

    In the \cc model~\cite{LPPP03,Lenzen13,PT11}  $n$ nodes (computers) are used to solve some problem.
    The input to the problem is spread among the nodes, and the communication among the nodes is done in \emph{rounds}.
	In each round, every node is allowed to send one message to every other node in the network.
	Typically, the size of every message is $O(\log n)$ bits, and messages to different nodes can be different.
	Usually, the input of every node is assumed to be of size $O(n)$ words, and so, can be sent to other nodes in one round, and one can hope for $O(1)$-round algorithms.
	In this model, the local computation time is ignored and the efficiency of an algorithm is measured by the number of \emph{communication rounds} made by the algorithm.

    One of the fundamental tasks in the \cc model is sorting of elements.
    In one of the seminal results for this model, Lenzen~\cite{Lenzen13} shows a sorting algorithm that run in $O(1)$ rounds.
    Lenzen's algorithm supports keys of size $O(\log n)$ bits.
    We show how to generalize Lenzen's sorting algorithm to support keys of size $O(n^{1-\eps})$ (for some constant $\eps$).
    Using this sorting algorithm we introduce efficient \cc algorithms for several string problems.

	\para{String sorting (\cref{sec:sorting_strings}).}
	The first algorithm is for the string sorting problem \cite{BS97,KR08,Bingmann18}.
	This is a special case of the large objects sorting problem, where the order defined on the objects is the lexicographical order.
    We introduce an $O(1)$ rounds algorithm for this specific order, even if there are strings of length $\omega(n)$.

	\para{Pattern matching (\cref{sec:pm}).}
	The second algorithm we present is an $O(1)$ rounds algorithm for pattern matching, which uses the string sorting algorithm.
	In the pattern matching problem the input is two strings, a pattern $P$ and a text $T$, and the goal is to find all the occurrences of $P$ in $T$.
	Algorithms for this problem were designed since the 1970's \cite{MP70,KMP77,KR87,Weiner73,BM77}.
	In the very related model of Massively Parallel Computing (MPC) (see discussion below), Hajiaghayi et al.~\cite{HSSS21} introduce a pattern matching algorithm that is based on FFT convolutions.
	Their algorithm can be adjusted to an $O(1)$ rounds algorithm in the \cc model.
	However, our algorithm has the advantage of using only combinatorial operations.

	\para{Suffix Array construction and the corresponding LCP array (\cref{sec:SA_and_LCP})}
	The last algorithm we present is an algorithm that constructs the suffix array~\cite{MM93,PST07} ($\SA$) of a given string, together with the corresponding longest common prefix ($\LCP$) array~\cite{KLAAP01,MM93}.
	The suffix array of a string $S$, denoted by $\SA_S$, is a sorted array of all of the suffixes of $S$.
	The $\LCP_S$ array stores for every two lexicographic successive suffixes the length of their longest common prefix.
	It was proved~\cite{MM93,KLAAP01} that the combination of $\SA_S$ with $\LCP_S$ is a powerful tool, that can simulate a bottom-up suffix tree traversal and is useful for solving problems like finding the longest common substring of two strings.
	Our algorithm takes $O(\log\log n)$ rounds.

\para{The input model.}
Most of the problems considered so far in the \cc model are graph problems.
For such problems where the input is a graph, it is very natural to consider partitioning of the input among $n$ nodes.
Each node in the network receives all the information on the neighborhood of one vertex of the input graph.
However, when the input is a set of objects or strings, like in our problems, it is less clear how the input is spread among the $n$ nodes of the \cc network.
We tried to minimize our assumption on the input to get as general algorithms as possible.
Inspired by the standard RAM model, we assume that the input of any problem can be considered as a long array that contains the input (just like the internal memory of the computer).
In the \cc model with $n$ nodes, we assume that the same input is now distributed among the local memories of the nodes (see~\cref{fig:input_model} for the case where the input objects are strings).
So, to get as an input a sequence of objects with a total size of $O(n^2)$ words, we consider their representation in a long array, one after the other, and then partition this array into $n$ pieces, each of length $O(n)$.
The assumption that the input of every node is of length $O(n)$, is consistent with previous problems considered in the \cc model~\cite{AKPR16,Lenzen13}.
A useful assumption we assume for the sake of simplicity is that when the size of every object is bounded by $O(n)$ words, any object is stored only in one node.
This assumption can be guaranteed within an overhead of $O(1)$ rounds.

\begin{figure}[t]
	\begin{center}
		\includegraphics[width=0.9\textwidth]{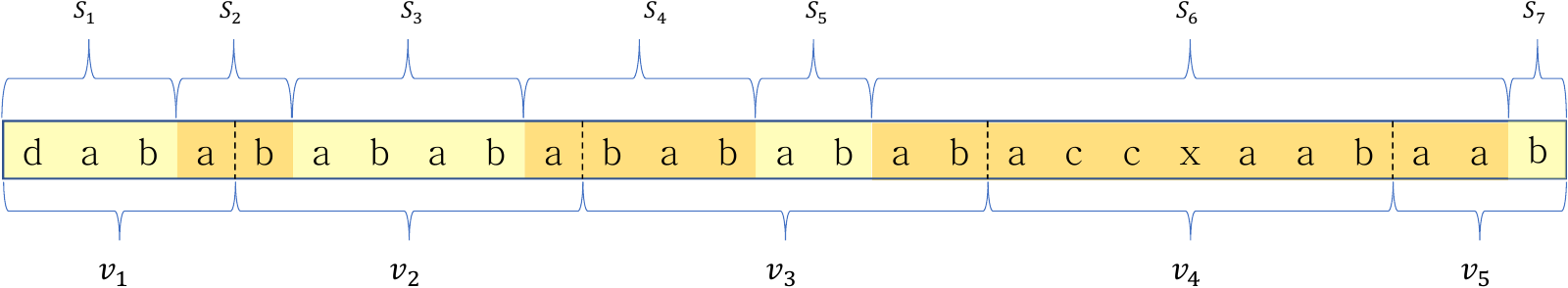}
	\end{center}
	\caption{An example of a sequence of strings, the input is partitioned between nodes $v_1,\dots,v_n$.}
	\label{fig:input_model}
\end{figure}

	\para{Relation between \cc and  MPC.}
	The Massively Parallel Computing (MPC) model~\cite{BKS13,ANOY14} is a very popular model that is useful for the analyzing of the more practical model of MapReduce~\cite{DG04}, from a theoretical point of view.
	In this model, a problem with input of size $O(N)$ words, is solved by $M$ machines, each with memory of size $S$ words such that $M\cdot S=\Theta(N)$.
	The MPC model is synchronous, and in every round each machine sends and receives information, such that the data received by each machine fits its memory size.
	We point out that, as described by Behnezhad et al.~\cite[{Theorem 3.1}]{BDH18}, every \cc algorithm with $\Theta(n^2)$ size input, that uses $O(n)$ space in every node, can be simulated in the MPC model, with $N=n^2$ and $S=\Theta(M)=\Theta(\sqrt N)$.
	Moreover, it is straightforward that every MPC algorithm that works with $S=\Theta(M)$ in $r$ rounds, can be simulated in an MPC instance with $S=\omega(M)$ (but still $M\cdot S=\Theta(N)$) in $r$ rounds since every machine can simulate several machines of the original algorithm.
	As a result, most of the algorithms we introduce in the \cc model implies also algorithms with the same round complexity in the MPC model for $S=\Omega(M)$.
	The only exception is the sorting algorithm for the case of $\eps=0$, which uses $\omega(n)$ memory in each machine (see \cref{sec:sorting_eps_0}).
	We note that the regime of $S=\Omega(M)$ is the most common regime for algorithms in the MPC model, see for example \cite{HSSS21,GGM22,Nowicki21}.

	\subsection{Related Work}

	\para{String Sorting.}
	The string sorting problem was studied in the PRAM model by Hagerup~\cite{Hagerup94}, where he introduces an optimal $O(\log n/\log\log n)$ time algorithm on a CRCW PRAM.
	The problem was also studied in the external I/O emory model by Arge et al.~\cite{AFGV97}.
	Recently, a more practical research was done by Bingman~\cite{Bingmann18}.

	\para{Pattern Matching.}
	In parallel models, on the 1980's, Galil~\cite{Galil85} and Vishkin~\cite{Vishkin85} introduced pattern matching algorithms in the CRCW and CREW PRAM models.
	Later, Breslauer and Galil~\cite{BG90} improved the complexity for CRCW from $O(\log n)$ to $O(\log\log n)$ rounds and show that this round complexity is optimal.

	\para{Suffix Array.}
	In the world of parallel and distributed computing, the problem of $\SA$ construction was studied in several models, both in theory~\cite{KSB06} and in practice~\cite{KS07,KKP15}.
	The most related line of work  is the results of K{\"{a}}rkk{\"{a}}inen et al.~\cite{KSB06} and the improvement for the Bulk Synchronous Parallel (BSP) model by Pace and Tiskin~\cite{PT13}.
	K{\"{a}}rkk{\"{a}}inen et al.~\cite{KSB06} introduce a linear time algorithm that works in several models of parallel computing, and requires $O(\log^2 n)$ synchronization steps in the BSP model (for the case of a polynomial number of processors).
	Their result uses a recursive process with a parameter that was used as a fixed value in all levels of recursion.
	Pace and Tiskin~\cite{PT13} show that one can enlarge the value of the parameter with the levels, what they called \emph{accelerated sampling}, such that the total work does not change asymptotically, but the depth of the recursion, and hence the number of synchronization steps, becomes $O(\log\log n)$.
	Our $\SA$ construction algorithm follows the same idea, but uses some different implementation details which fit the \cc model, and exploits our large-objects sorting algorithm.

\subsection{Our Contribution}
    Our results are summarized in the following theorems:

    	\begin{theorem}[String Sorting]\label{thm:string_sorting}
		There is an algorithm that given a sequence of strings $\S=(S_1,S_2,\dots,S_k)$, computes $\rank_{\S}(S_j)$ for every string $S_j\in \S$, and stores this rank in the node that receives $S_j[1]$.
		The running time of the algorithm is $O(1)$ rounds of the \cc model.
	\end{theorem}

	\begin{theorem}[Pattern Matching]\label{thm:pm}
		There is an algorithm that given two strings $P$ and $T$, computes for every $i\in[0..|T|-|P|+1]$ whether $T[i+1..i+|P|]=P$ in $O(1)$ rounds of the \cc model.
	\end{theorem}

	\begin{theorem}[Suffix Array and $\LCP$]\label{thm:SA_LCP}
		There is an algorithm that given a string $S$, computes $\SA_S$ and $\LCP_S$ in $O(\log\log n)$ rounds of the \cc model.
	\end{theorem}

	As described above, all our results are based on the algorithm for sorting large objects.
    The problem is defined as follows
	\begin{problem}[{Large Object Sorting}]\label{prob_sorting}
		Assume that a \cc network of $n$ nodes gets a sequence $\B=(B_1,B_2,\dots,B_k)$ of objects, each of size $O(n^{1-\eps})$ words for $\eps\ge 0$, where the total size of $\B$'s objects is $O(n^2)$ words.
		For every object $B_j\in \B$, the node that gets $B_j$ needs to learn  $\rank_{\B}(B_j)$.
	\end{problem}

The algorithms for \cref{prob_sorting} are presented in \cref{thm:sorting_esp_g0,thm:sorting_esp_0}, which we prove in \cref{sec:sorting_not_too_large_objects,sec:sorting_eps_0}, respectively.

	\begin{theorem}\label{thm:sorting_esp_g0}
		There is an algorithm that solves~\cref{prob_sorting} for any constant $\eps>0$ in $O(1)$ rounds of the \cc model.
	\end{theorem}

	\begin{theorem}\label{thm:sorting_esp_0}
		There is an algorithm that solves~\cref{prob_sorting} for $\eps=0$ in $O(\log n)$ rounds.
		Moreover, any comparison-based \cc algorithm that solves~\cref{prob_sorting} for $\eps=0$ requires $\Omega(\log n/\log\log n)$ rounds.
	\end{theorem}

	\section{Preliminaries}\label{sec:preliminaries}

	For $i,j\in\mathbb{Z}$ we denote $[i..j]=\{i,i+1,i+2,\dots,j\}$. For a set $S\subseteq \mathbb Z$ and a scalar $\alpha\in\mathbb Z$ we denote $S+\alpha=\{a+\alpha\mid a\in S\}$.
	For a set $\K$ of elements with a total order, and an element $b\in \K$ we denote $\rank_\K(b)=|\{a\in \K\mid a<b\}|$ (or simply $\rank(b)$ when $\K$ is clear from the context).
	We clarify that for a multi-set $M$, we consider the rank of an element $b\in M$  to be the number of \emph{distinct} elements smaller than $b$ in $M$.
	For a set of objects $\K$ we denote $\|\K\| = \sum_{B\in\K}|B|$ as the total size (in words of space) of the objects in $\K$.

	\para{Strings.}
	A string $S=S[1]S[2]\dots S[n]$ over an alphabet $\Sigma$ is a sequence of characters from $\Sigma$.
	In this paper we assume $|\Sigma|=[1..\mathsf{poly}(n)]$ and therefore each character takes $O(\log n)$ bits.
	The length of $S$ is denoted by $|S|=n$.
	For $1\le i<j\le n$ the string $S[i..j]=S[i]S[i+1]\dots S[j]$ is called a substring of $S$. if $i=1$ then $S[i..j]$ is called a prefix of $S$ and if $j=n$ then $S[i..j]$ is a suffix of $S$ and is also denoted as $S[i..]$.
	The following lemma from \cite{BG14} is useful for our pattern matching algorithm in \cref{sec:pm}.
	\begin{lemma}[{\cite[{Lemma 3.1}]{BG14}}]\label{lem:BG_progression}
		Let $u$ and $v$ be two strings such that $v$ contains at least three occurrences	of $u$. Let $t_1 < t_2 < \cdots < t_h$ be the locations of all occurrences $u$ in $v$ and assume that  $t_{i+2}-t_i\le |u|$, for $i=[1..h-2]$ and $h \ge 3$.
		Then, this sequence forms an arithmetic progression with difference $d=t_{i+1}-t_i$, for $i=[1..h-1]$ (that is equal to the period	length of $u$).
	\end{lemma}

 Here is the definition of the longest common prefix of two strings. It is useful for the definition of the lexicographical order and also for the LCP array.
	\begin{definition}
		For two strings $S_1,S_2\in\Sigma^*$, we denote $\LCP(S_1,S_2)=$
		$\max(\{\ell\mid S_1[1..\ell]=S_2[1..\ell]\}\cup\{0\})$ to be the length of the longest common prefix of $S_1$ and $S_2$.
	\end{definition}

	We provide here the formal definition of \emph{lexicographic order} between strings.
	\begin{definition}[{Lexicographic order}]\label{def:lex_order}
		For two strings $S_1,S_2\in\Sigma^*$ we have $S_1\preceq S_2$ if one of the following holds:
		\begin{enumerate}
			\item If $\ell=\LCP(S_1,S_2)<\min\{|S_1|,|S_2|\}$ and $S_1[\ell+1]<S_2[\ell+1]$
			\item $S_1$ is a prefix of $S_2$, i.e. $|S_1|\le |S_2|$ and $S_1=S_2[1..|S_1|]$.
		\end{enumerate}
		We denote the case where $S_1\preceq S_2$ and $S_1\ne S_2$ as $S_1 \prec S_2$.
	\end{definition}

	\para{Routing.}

    In the Congested Clique model, a routing problem involves delivering messages from a set of source nodes to a set of destination nodes, where each node may need to send and receive multiple messages.
    A well-known result by Lenzen~\cite{Lenzen13} shows that if each node is the source and destination of at most $O(n)$ messages, then all messages can be delivered in $O(1)$ rounds.
    The following lemma is useful for routing in the \cc model.
		\begin{lemma}[{\cite[{Lemma 9}]{CFG20}}]\label{thm:gen_routing}
		Any routing instance, in which every node $v$ is the target of up to $O(n)$ messages, and $v$ locally computes the messages it desires to send from at most $O(n \log n)$ bits, can be performed in O(1) rounds.
	\end{lemma}

	\begin{remark*}
        A particularly useful case in which $v$ locally computes the messages it wishes to send from $O(n \log n)$ bits is when $v$ stores $O(n)$ messages, each intended for all nodes in some consecutive range of nodes in the network.

    We also provide a routing lemma that consider a symmetric case to that of \cref{thm:gen_routing}, which we prove in \cref{sec:queries_routing}.
        \begin{lemma}\label{lem:queries_routing}
        Assume that each node of the \cc stores $O(n)$ words of space, each of size $O(\log n)$.
		Each node has $O(n)$ queries, such that each query is a pair of a resolving node, and the content of the query which is encoded in $O(\log n)$ bits.
		Moreover, the query can be resolved by the resolving node, and the size of the result is $O(\log n)$ bits.
		Then, it is possible in $O(1)$ rounds of the \cc that each node get all the results of its queries.
	\end{lemma}

	\end{remark*}

    	\section{Sorting Large Objects}\label{sec:sorting_not_too_large_objects}
	In this section, we solve \cref{prob_sorting}, the large objects sorting problem, for the special case of $\eps=2/3$, in the \cc model, by presenting a deterministic sorting algorithm for objects of size $O(n^{1/3})$ words, that takes $O(1)$ rounds.
	Later, in \cref{sec:sorting_general_large_objects}, we generalize the algorithm for any $\eps>0$, which proves \cref{thm:sorting_esp_g0}.

    Our algorithm makes use of the following two lemmas.

	\begin{lemma}[{\cite[{Lemma 3}]{CDKL21}}]\label{lem:dist_info}
		Let $x_1,x_2,\dots, x_n$ be natural numbers, and let $X$, $x$ and $k$ be natural numbers such that $\sum_{i=1}^{n}  x_i = X$, $x_i \le x$ for all $i$.
		Then there is a partition of $[n]$ into $k$ sets $I_1, I_2,\dots, I_k$ such
		that for each $j$, the set $I_j$ consists of consecutive elements, and
		$\sum_{i\in I_j}x_i\le X/k+x$.
	\end{lemma}

\begin{lemma}[{\cite[{Lemma 1.2}]{JN18}}]\label{lem:aux_ndes}
	Let $A$ be a \cc algorithm which, except of  the  nodes $u_1,\dots , u_n$ corresponding  to the  input strings,  uses $O(n)$ auxiliary nodes $v_1, v_2, \dots $ such that  the auxiliary nodes do not have initially any knowledge of the input  strings  on  the  nodes $u_1,\dots, u_n$.  Then, each round of $A$ might be simulated in $O(1)$ rounds in the  standard \cc model, without auxiliary nodes.
\end{lemma}

	Our algorithm is a generalization of Lenzen's~\cite{Lenzen13} sorting algorithm. In~\cite{Lenzen13}, each node is given $n$ keys of size $O(1)$ words of space (i.e. $O(\log n)$ bits) and the nodes need to learn the ranks of their keys in the total order of the union of all keys.
	Our algorithm uses similar methods but has another level of recursion to handle also objects of size $\omega(1)$  (yet $O(n^{1/3})$) words.

	In this section, we prove the following lemma.

	\begin{lemma}\label{lem:sorting_objects}
		Consider a variant of~\cref{prob_sorting} where each object is of size $O(n^{1/3})$ words.
		Moreover, every object is stored in one node.
		Then, there exists an algorithm that solves this variant in $O(1)$ rounds.
	\end{lemma}

	The main part of the algorithm is sorting the objects of the network by redistributing the objects among the nodes such that for any two objects $B< B'$ that the algorithm sends to nodes $v_i$ and $v_j$, respectively, we have $i\le j$.

	The algorithm stores with each object $B$ the original node of $B$ and the index of $B$ in the original node.
	The algorithm uses the order of original nodes and indices to break ties.

	To sort large objects of size $O(n^{1/3})$ words, the algorithm uses two building blocks.
	First, we show how to sort the objects of a set of $n^{1/3}$ nodes\footnote{We assume that $n^{1/3}$ is an integer. Otherwise, we add $O(n)$ auxiliary nodes such that the total number of nodes $n'$ holds $n'^{1/3}\in \mathbb{N}$. By~\cref{lem:aux_ndes} the round complexity is increased only by a constant factor.}.
	This algorithm is the base of the second building block, which is a recursive algorithm that sorts the objects of a set of $\omega(n^{1/3})$ nodes.

	\subsection{Sorting at most \texorpdfstring{$n^{1/3}$}{n\textasciicircum 1/3} Nodes with Objects of Size \texorpdfstring{$O(n^{1/3})$}{O(n\textasciicircum 1/3)}}

	In this section we present \cref{alg:n13}, that sorts all the objects in a set of nodes $W\subset V$ of at most $n^{1/3}$ nodes, each object of size $O(n^{1/3})$ words, and each node stores $O(n)$ words.
	As in~\cite{Lenzen13}, each node marks some objects as \emph{candidates}.
	Then, $n^{1/3}$ of the candidates are chosen to be delimiters, and the objects are redistributed according to these \emph{delimiters}.
	The main part of the analysis is to prove that the redistribution works well, i.e. the delimiters divide the nodes into sets of almost evenly sizes and therefore each set can be sent to one node in $O(1)$ rounds.

	\begin{algorithm}[h!]
		\caption{Sorting objects of at most $n^{1/3}$ nodes}\label{alg:n13}
		\KwIn{Set $W$ of at most $n^{1/3}$ nodes, each node stores objects of size $O(n^{1/3})$ words each, a total of $O(n)$ words per node and every object is stored in one node.}
		Each node in $W$ locally sorts its objects\label{algn13:ln1}\;
		Each node in $W$ marks for every positive integer $i$ the smallest (due to the order of step~\ref{algn13:ln1}) object $B$ such that the total size of all the objects smaller than $B$ is at least $i\cdot n^{2/3}$. The marked objects are called candidates\label{algn13:ln2}\;
		Each node in $W$ announces the candidates to all other nodes in $W$\label{algn13:ln3}\;
		Let $\T$ be the union of the candidates. Each node in $W$ locally sorts $\T$ and selects every $\ceil{|\T|/|W|}$th object according to this order.
		We call such an object a delimiter\label{algn13:ln4}\;
		Each node $v_i\in W$ splits its original input into $|W|$ subsets, where the $j$th subset $K_{i,j}$ contains all objects that are larger than the $(j-1)$th delimiter (for $j = 1$ this condition does not apply) and smaller or equal to the $j$th delimiter (for $j = |W|$ this condition does not apply)\label{algn13:ln5}\;
		Each node $v_i\in W$ sends $K_{i,j}$ to the $j$th node of $W$\label{algn13:ln6}\;
		Each node $v_i$ in $W$ locally sorts the objects $v_i$ receives in~\ref{algn13:ln6}\label{algn13:ln7}\;
	\end{algorithm}

	\para{Correctness.}The correctness of~\cref{alg:n13} derives from steps~\ref{algn13:ln4} to~\ref{algn13:ln6}. As in \cite[{Lemma 4.2}]{Lenzen13} due to the partitioning by delimiters, all the objects in $K_{i,j}$ are larger than the objects in $K_{i',j'}$ for all $v_i,v_{i'} \in W$ and $j' < j$.

	\para{Complexity.}We now show that~\cref{alg:n13} runs in $O(1)$ rounds.
	Notice that communication only happens in steps~\ref{algn13:ln3} and~\ref{algn13:ln6}.
	In both steps, each node sends $O(n)$ words.
	We will show that each node also receives $O(n)$ words, and therefore we can use Lenzen's routing scheme.

	For step~\ref{algn13:ln3}, notice that $|W|\le n^{1/3}$ nodes, there are $O(n^{1/3})$ candidates per node, and the size of any candidate is $O(n^{1/3})$ words.
	Therefore each node receives $O(n^{1/3})\cdot O(n^{1/3})\cdot O(n^{1/3})=O(n)$ words of space.

	It is left to prove that in step~\ref{algn13:ln6} each node receives $O(n)$ words.
	A similar argument was also proved in \cite[{Lemma 4.3}]{Lenzen13}, and we show here that partitioning the objects using step~\ref{algn13:ln2} is an efficient interpretation of Lenzen's sorting algorithm in the case of large objects.

	\begin{lemma}\label{lem:unionpiecebasic}
		When executing~\cref{alg:n13}, for each $j\in [1..|W|]$, it holds that 	\\$\left\|\bigcup_{i=1}^{|W|}K_{i,j}\right\|=O(n)$.
	\end{lemma}

 	\begin{proof}
        Let $d_i$ be the number of candidates in $K_{i,j}$.
	First notice that due to the choice of the delimiters, $\bigcup_{i=1}^{|W|}K_{i,j}$ contains at most $\ceil{|\T|/|W|} = |W|\cdot O(n^{1/3})/|W| = O(n^{1/3})$
		candidates and therefore $\sum_{i=1}^{|W|}d_i=\ceil{|\T|/n^{1/3}}=O(n^{1/3})$.

		Since by step~\ref{algn13:ln2} the total size of objects between two consecutive candidates in one node is $O(n^{2/3})$ words, we have that $\|K_{i,j}\| \le (d_i+1)\cdot O(n^{2/3})$.

		Therefore,
            \[\mathop{\left\|{\bigcup}_{i=1}^{|W|}K_{i,j}\right\|} = \sum_{i=1}^{|W|}\|K_{i,j}\| = O(n^{2/3})\cdot \sum_{i=1}^{|W|}	(d_i+1)
		O(n^{2/3})\cdot O(n^{1/3}+|W|) = O(n).\]
	\end{proof}

	\subsection{Sorting more than \texorpdfstring{$n^{1/3}$}{n\textasciicircum 1/3} Nodes with Objects of Size \texorpdfstring{$O(n^{1/3})$}{O(n\textasciicircum 1/3)}}

	The following algorithm sorts all the objects in $U\subseteq V$ for $|U|>n^{1/3}$ nodes, where each object is of size $O(n^{1/3})$ words, and each node stores $O(n)$ words.
	In particular, this algorithm sorts all the objects in $n$ nodes.

	In this recursive algorithm, each node marks some objects as candidates, such that all the candidates are fit into $O({|U|/n^{1/3}})$ nodes.
	The candidates are sorted recursively, and $n^{1/3}$ of the candidates are chosen to be delimiters.
	Then, the objects are redistributed according to these delimiters, and each set of size $O(|U|/n^{1/3})$ nodes is sorted recursively.

	\begin{algorithm}[h!]
		\caption{Sorting objects of more than $n^{1/3}$ nodes}\label{alg:n23}
		\KwIn{Set $U$ with  $|U|>n^{1/3}$ nodes, each node stores objects of size $O(n^{1/3})$ words each, a total of $O(n)$ words per node and every object is stored in one node.}

		Each node in $U$ locally sorts its objects\label{algn23:ln1}\;

		Each node in $U$ marks for every positive integer $i$ the smallest (due to the order of step~\ref{algn13:ln1}) object $B$ such that the total size of all the objects smaller than $B$ is at least $i\cdot n^{2/3}$. The marked objects are called candidates\label{algn23:ln2}\;

		All the candidates are distributed among the first $\ceil{|U|/n^{1/3}}$ nodes (see details below)\label{algn23:ln3}\;

		Using~\cref{alg:n13} (if $\ceil{|U|/n^{1/3}}\leq n^{1/3}$) or~\cref{alg:n23} (otherwise), the first $\ceil{|U|/n^{1/3}}$ nodes sort all the candidates\label{algn23:ln4}\;

		Let $\T$ be the union of the sorted candidates in the first $\ceil{|U|/n^{1/3}}$ nodes.
		Every $\ceil{|\T|/n^{1/3}}$th object according to this order is selected to be a delimiter (see details below).
		The delimiters are announced to all the nodes in $U$\label{algn23:ln5}\;

		Each node $v_i\in U$ splits its original input into $n^{1/3}$ subsets,
		where the $j$th subset $K_{i,j}$ contains all objects that are larger than the $(j-1)$th delimiter
		(for $j = 1$ this condition does not apply) and smaller or equal to the $j$th delimiter (for $j = n^{1/3}$ this condition does not apply)\label{algn23:ln6}\;
		The nodes of $U$ are partitioned into $n^{1/3}$ disjoint sets $W_1,W_2,\dots,W_{n^{1/3}}$, each of size $\floor{|U|/n^{1/3}}$ nodes.
		Each node $v_i\in U$ sends $K_{i,j}$ to $W_j$ (see details below)\label{algn23:ln7}\;
		Using~\cref{alg:n13} (if $\floor{|U|/n^{1/3}}\leq n^{1/3}$) or~\cref{alg:n23} (otherwise), each set $W_j$ sorts all the objects received in $W_j$\label{algn23:ln8}\;
	\end{algorithm}

	\para{Correctness.} The correctness of~\cref{alg:n23} stems from steps~\ref{algn23:ln5} to~\ref{algn23:ln8} and follows analogously to the correctness of~\cref{alg:n13}.

	\para{Complexity.} We will focus on the steps in~\cref{alg:n23} where communication is made and show that each step takes $O(1)$ rounds.

	In step~\ref{algn23:ln3}, we need to further explain some algorithmic details.
	Each node sends $O(n^{1/3})$ objects of size $O(n^{1/3})$ words, so at most $O(n^{2/3})$ words per node.
	The candidates of node $v_i$ are sent to node $v_j$ for $j=\ceil{\frac{i}{n^{1/3}}}$, therefore each node receives at most $n^{1/3}\cdot O(n^{2/3})=O(n)$ words.
	By Lenzen's routing scheme
	this is done in $O(1)$ rounds.
	In step~\ref{algn23:ln4}, notice that since the number of nodes is at most $n$, the depth of the recursion is $O(1)$, therefore the recursion does not increase the round complexity asymptotically.

	In step~\ref{algn23:ln5}, the delimiters should be recognized.
	Each node $v$ in the first $\ceil{|U|/n^{1/3}}$ nodes broadcasts the number of objects that $v$ receives in step~\ref{algn23:ln4}.
	Therefore, each node $v$ computes for every object $B$ whether the rank of $B$ among the candidates is a multiple of $\ceil{|\T|/n^{1/3}}$ and if so, selects $B$ to be a delimiter.
	There are $O(n^{1/3})$ delimiters of size $O(n^{1/3})$ each, which is in total $O(n^{2/3})$ words. Therefore, the delimiters are announced to all the nodes in $O(1)$ rounds using \cref{thm:gen_routing}.

	In step~\ref{algn23:ln7} we first show that the total number of words that each set $W_j$ receives, is $O(|U|\cdot n^{2/3})$ words.

    \begin{lemma}\label{lem:unionpieceadv}
        When executing~\cref{alg:n23}, for each $j\in [1..n^{1/3}]$, it holds that
		\\$ \left\| \bigcup_{i=1}^{|U|}K_{i,j}\right\|=O(|U|\cdot n^{2/3})$.
	\end{lemma}

		\begin{proof}
		The proof is similar to the proof of~\cref{lem:unionpiecebasic}.
		Let $d_i$ be the number of candidates in $K_{i,j}$.
		Due to the choice of the delimiters, $\bigcup_{i=1}^{|U|}K_{i,j}$ contains $\ceil{|\T|/n^{1/3}}=|U|\cdot O(n^{1/3})/n^{1/3}= O(|U|)$ candidates and therefore $\sum_{i=1}^{|U|}d_i=\ceil{|\T|/n^{1/3}}=O(|U|)$.
		Since by step~\ref{algn23:ln2} the total size of objects between two consecutive candidates is $O(n^{2/3})$ words, we have that $\|K_{i,j}\| \le (d_i+1)\cdot O(n^{2/3})$.
		The number of words that are sent to set $W_j$ is at most
	\begin{align*}
 \left\|\bigcup_{i=1}^{|U|}{K_{i,j}}\right\|&=\sum_{i=1}^{|U|}{\|K_{i,j}\|}= O(n^{2/3})\cdot \sum_{i=1}^{|U|} (d_i+1)
  \\&=O(n^{2/3})O(|U|+|U|)=O(n+n^{2/3}|U|)=O(|U|\cdot n^{2/3}).
  \end{align*}
	\end{proof}

	Now, the algorithm selects for every set $W_j$ a leader $v_{W_j}$.
	Each node $v_i$ sends $\|K_{i,j}\|$ to $v_{W_j}$.
	The leader $v_{W_j}$ computes and sends to each node $v_i\in U$ a node $w\in W_j$ such that $v_i$ should send $K_{i,j}$ to $w$.
	By \cref{lem:dist_info}, there is a computation such that for each $w\in W_j$, the number of words that $w$ receives is at most $O(n)$ words,
	by setting in the lemma $x_i\coloneqq\|K_{i,j}\|$, $X\coloneqq|U|\cdot O(n^{2/3})$, $x\coloneqq O(n)$ and $k\coloneqq|W_j|=\floor{|U|/n^{1/3}}$.
	On the other hand, each node sends at most $O(n)$ words.
	By Lenzen's routing scheme this is done in $O(1)$ rounds.

	In step~\ref{algn23:ln8}, similar to step~\ref{algn23:ln4}, the depth of the recursion is $O(1)$.

	We are ready to prove~\cref{lem:sorting_objects}.
	\begin{proof}[{Proof of~\cref{lem:sorting_objects}}]
		First, we apply~\cref{alg:n23} with $U=V$.
		Hence, all the objects are ordered in a non-decreasing lexicographical order among all the nodes of the network.

		Next, we show how to compute for each object $B$, $\rank(B)$.
		Each node $v_i$ for $1\leq i< n$ sends to node $v_{i+1}$ the largest object of $v_i$ (by the lexicographical order), denoted $B^\ell_i$.
		Then, each node $v_i$ computes and broadcasts the number of distinct objects $v_i$ holds that are different from $B^\ell_{i-1}$ (for $i=1$, node $v_i$ just broadcasts the number of distinct objects $v_i$ holds), ignoring the tiebreakers of the original node and original index (notice that the number of distinct objects might be 0).
		Now, each node $v_i$ computes the rank of all the objects $v_i$ holds.

		Lastly, for every object $B$, $\rank(B)$ is sent to the original node of $B$, using the information of the original node that $B$ stores.
		By Lenzen's routing scheme
		this is done in $O(1)$ rounds.
	\end{proof}

\subsection{Sorting Objects of Size \texorpdfstring{$O(n^{1-\eps})$}{O(n\textasciicircum (1-ε))}}\label{sec:sorting_general_large_objects}
	In this section we explain how to sort objects of size $O(n^{1-\eps})$ for a constant $\eps>0$.
	The algorithm is a straightforward generalization of \cref{alg:n13} to general ${1-\eps}$ instead of $1/3$ (i.e. $\eps = 2/3$).

	First, in \cref{alg:n1e} we show an algorithm that sorts $n^{\eps/2}$ nodes.
	\begin{algorithm}[h!]
		\caption{Sorting objects of at most $n^{\eps/2}$ nodes}\label{alg:n1e}
		\KwIn{Set $W$ of at most $n^{\eps/2}$ nodes, each node stores objects of size $O(n^{1-\eps})$ words each, a total of $O(n)$ words per node and every object is stored in one node.}
		Each node in $W$ locally sorts its objects\label{algn1e:ln1}\;
		Each node in $W$ marks for every positive integer $i$ the smallest (due to the order of step~\ref{algn13:ln1}) object $B$ such that the total size of all the objects smaller than $B$ is at least $i\cdot n^{1-\eps/2}$. The marked objects are called candidates\label{algn1e:ln2}\;
		Each node in $W$ announces the candidates to all other nodes in $W$\label{algn1e:ln3}\;
		Let $\T$ be the union of the candidates. Each node in $W$ locally sorts $\T$ and selects every $\ceil{|\T|/|W|}$th object according to this order.
		We call such an object a \emph{delimiter}\label{algn1e:ln4}\;
		Each node $v_i\in W$ splits its original input into $|W|$ subsets, where the $j$th subset $K_{i,j}$ contains all objects that are larger than the $(j-1)$th delimiter (for $j = 1$ this condition does not apply) and smaller or equal to the $j$th delimiter (for $j = |W|-1$ this condition does not apply)\label{algn1e:ln5}\;
		Each node $v_i\in W$ sends $K_{i,j}$ to the $j$th node in $W$\label{algn1e:ln6}\;
		Each node $v_i$ in $W$ locally sorts the objects $v_i$ received in~\ref{algn1e:ln6}\label{algn1e:ln7}\;
	\end{algorithm}

	The correctness and complexity follows analogously with the correctness and complexity of \cref{alg:n13}.

	\para{Correctness.}The correctness of~\cref{alg:n13} derives from steps~\ref{algn13:ln4} to~\ref{algn13:ln6}.
	As in \cite[{Lemma 4.2}]{Lenzen13} due to the partitioning by delimiters, all the objects in $K_{i,j}$ are larger than the objects in $K_{i',j'}$ for all $v_i,v_i' \in W$ and $j' < j$.

	\para{Complexity.}We now show that~\cref{alg:n1e} runs in $O(1)$ rounds.
	Notice that communication only happens in steps~\ref{algn1e:ln3},  \ref{algn1e:ln6}.
	In both steps, each node sends $O(n)$ words.
	We will show that each node also receives $O(n)$ words, and therefore we can use Lenzen's routing scheme.

	For step~\ref{algn1e:ln3}, notice that there are at most $n^{\eps/2}$ nodes, $O(n^{\eps/2})$ candidates per node, and $O(n^{1/3})$ size per object.
	Therefore each node receives $O(n^{\eps/2})\cdot O(n^{\eps/2})\cdot O(n^{1-\eps})=O(n)$ words of space.

	It is left to prove that in step~\ref{algn1e:ln6} each node receives $O(n)$ words.

	\begin{lemma}\label{lem:unionpiecebasic_eps}
		When executing~\cref{alg:n1e}, for each $j\in [1..|W|]$, it holds that
		\[\left
        \|\bigcup_{i=1}^{|W|}K_{i,j}\right\|=O(n).\]
	\end{lemma}

	\begin{proof}
		Let $d_i$ be the number of candidates in $K_{i,j}$.
		First notice that due to the choice of the delimiters, $\bigcup_{i=1}^{|W|}K_{i,j}$ contains $\ceil{|\T|/|W|} = |W|\cdot O(n^{\eps/2})/|W| = O(n^{\eps/2})$ candidates and therefore $\sum_{i=1}^{|W|}d_i=\ceil{|\T|/|W|}=O(n^{\eps/2})$.
		Since by step~\ref{algn13:ln2} the total size of objects between two consecutive candidates is $O(n^{1-\eps/2})$ words, we have that $\|K_{i,j}\| \le (d_i+1)\cdot O(n^{1-\eps/2})$.

		Therefore,
		\[\mathop{\left\|\underset{v_i\in W}{\bigcup}K_{i,j}\right\|} = \sum_{v_i\in W}\|K_{i,j}\| = O(n^{1-\eps/2})\cdot \underset{v_i\in W}{\sum}(d_i+1)
		= O(n^{1-\eps/2})\cdot O(n^{\eps/2}+|W|)
		= O(n).\]
	\end{proof}

	Next, in \cref{alg:n2e} we show an algorithm for more than $n^{\eps/2}$ nodes.

	\begin{algorithm}[h!]
		\caption{Sorting objects of $\omega(n^{\eps/2})$ nodes}\label{alg:n2e}
		\KwIn{Set $U$ with  $|U|>n^{\eps/2}$ nodes, each node stores objects of size $O(n^{1-\eps})$ words each, a total of $O(n)$ words per node and every object is stored in one node.}

		Each node in $U$ locally sorts its objects\label{algn2e:ln1}\;

		Each node in $U$ marks for every positive integer $i$ the smallest (due to the order of step~\ref{algn13:ln1}) object $B$ such that the total size of all the objects smaller than $B$ is at least $i\cdot n^{1-\eps/2}$. The marked objects are called candidates\label{algn2e:ln2}\;

		All the candidates are distributed among the first $\ceil{|U|/n^{\eps/2}}$ nodes\label{algn2e:ln3}\;
		Using~\cref{alg:n1e} (if $|U|/n^{\eps/2}\leq n^{\eps/2}$) or~\cref{alg:n2e} (otherwise), the first $\ceil{|U|/n^{\eps/2}}$ nodes sort all the candidates\label{algn2e:ln4}\;

		Let $\T$ be the union of the sorted candidates in the first $\ceil{|U|/n^{\eps/2}}$ nodes.
		Every $\ceil{|\T|/(n^{\eps/2})}$th object according to this order is selected to be a delimiter.
		The delimiters are announced to all the nodes in $U$\label{algn2e:ln5}\;

		Each node $v_i\in U$ splits its original input into $n^{\eps/2}$ subsets,
		where the $j$th subset $K_{i,j}$ contains all objects that are larger than the $(j-1)$th delimiter
		(for $j = 1$ this condition does not apply) and smaller or equal to the $j$th delimiter (for $j = n^{\eps/2}-1$ this condition does not apply)\label{algn2e:ln6}\;
		The nodes of $U$ are partitioned into $n^{\eps/2}$ disjoint sets $\mathcal{W}$ of size $\floor{|U|/n^{\eps/2}}$ nodes in each set.
		Each node $v_i\in U$ sends $K_{i,j}$ to $W_j$, the $j$th set of $\mathcal{W}$ \label{algn2e:ln7}\;
		Using~\cref{alg:n1e} (if $|U|/n^{\eps/2}\leq n^{\eps/2}$) or~\cref{alg:n2e} (otherwise), each set $W_j\in \mathcal{W}$ sorts all the objects received in $W_j$\label{algn2e:ln8}\;
	\end{algorithm}

	Again, 	the correctness and complexity follows analogously with the correctness and complexity of \cref{alg:n23}.

	Notice that the running time of \cref{alg:n2e} is $T(|U|) = 2\cdot T(|U|/n^{\eps/2}) + O(1)$. Since $|U|\leq n$, there are at most $O(2^{2/\eps})$ rounds, which are $O(1)$ rounds for constant $\eps>0$.

	We now analyze the complexity of \cref{alg:n2e} for the sake of completeness.

	\para{Complexity.} We will focus on the steps in~\cref{alg:n2e} where communication is made and show that each step takes $O(1)$ rounds.

	In step~\ref{algn2e:ln3}, we need to further explain some algorithmic details.
	Each node sends $O(n^{1-\eps})$ objects of size $O(n^{\eps/2})$ words, so at most $O(n^{1-\eps/2})$ words per node.
	The candidates of node $v_i$ are sent to node $v_j$ for $j=\ceil{\frac{i}{n^{\eps/2}}}$, therefore each node receives at most $n^{\eps/2}\cdot O(n^{1-\eps/2})=O(n)$ words.
	By Lenzen's routing scheme
	this is done in $O(1)$ rounds.
	In step~\ref{algn2e:ln4}, notice that since the number of nodes is at most $n$, the depth of the recursion is $O(1)$, therefore the recursion does not increase the round complexity asymptotically.

	In step~\ref{algn2e:ln5}, the delimiters should be recognized.
	Each node $v$ in the first $\ceil{|U|/n^{\eps/2}}$ nodes broadcasts the number of objects that $v$ receives in step~\ref{algn2e:ln4}.
	Therefore, each node $v$ computes for every object $B$ whether the rank of $B$ among the candidates is a multiple of $\ceil{|\T|/n^{\eps/2}}$ and if so, selects $B$ to be a delimiter.
	There are $O(n^{\eps/2})$ delimiters of size $O(n^{1-\eps})$ each, which is in total $O(n^{1-\eps/2})$ words. Therefore, the delimiters are announced to all the nodes in $O(1)$ rounds using \cref{thm:gen_routing}.

	In step~\ref{algn2e:ln7} we first show that the total number of words that each set $W_j$ receives, is $O(|U|\cdot n^{1-\eps/2})$ words.

	\begin{lemma}\label{lem:unionpieceadv_e}
		When executing~\cref{alg:n2e}, for each $j\in [1..n^{\eps/2}]$, it holds that
		\[ \left\|\bigcup_{i=1}^{|U|}K_{i,j}\right\|=O(|U|\cdot n^{1-\eps/2}).\]
	\end{lemma}

	\begin{proof}
		Let $d_i$ be the number of candidates in $K_{i,j}$.
		Due to the choice of the delimiters, $\bigcup_{i=1}^{|U|}K_{i,j}$ contains $\ceil{|\T|/n^{\eps/2}}=|U|\cdot O(n^{\eps/2})/n^{\eps/2}= O(|U|)$ candidates and therefore $\sum_{i=1}^{|U|}d_i=\ceil{|\T|/n^{\eps/2}}=O(|U|)$.
		Since by step~\ref{algn2e:ln2} the total size of objects between two consecutive candidates is $O(n^{1-\eps/2})$ words, we have that $\|K_{i,j}\| \le (d_i+1)\cdot O(n^{1-\eps/2})$.
		The number of words that are sent to set $W_j$ is at most
		\[\left\|\bigcup_{i=1}^{|U|}{K_{i,j}}\right\|=\sum_{i=1}^{|U|}{\|K_{i,j}\|}= O(n^{1-\eps/2})\cdot \sum_{i=1}^{|U|} (d_i+1)=O(n^{1-\eps/2})O(|U|+|U|)=O(|U|\cdot n^{1-\eps/2}).\]
	\end{proof}

	Now, the algorithm selects for every set $W_j$ a leader $v_{W_j}$.
	Each node $v_i$ sends $\|K_{i,j}\|$ to $v_{W_j}$.
	The leader $v_{W_j}$ computes and sends to each node $v_i\in U$ a node $w\in W_j$ such that $v_i$ should send $K_{i,j}$ to $w$.
	By \cref{lem:dist_info}, there is a computation such that for each $w\in W_j$, the number of words that $w$ receives is at most $O(n)$ words.
	On the other hand, each node sends at most $O(n)$ words.
	By Lenzen's routing scheme this is done in $O(1)$ rounds.

	In step~\ref{algn2e:ln8}, similar to step~\ref{algn2e:ln4}, the depth of the recursion is $O(1)$.

 \begin{remark}[$\eps\in o(1)$ case]
 \label{remark:small_eps}
     For $\eps\in o(1)$, the depth of the recursion in steps \ref{algn2e:ln4} and \ref{algn2e:ln8}, is $2/\eps$.
     Moreover, the algorithm performs two calls to the recursive process.
     Therefore, the running time of this algorithm is $O(2^{2/\eps})$ steps.
 \end{remark}

	\section{Sorting Strings}\label{sec:sorting_strings}
	Although in \cref{sec:sorting_eps_0} we showed that it is impossible to sort general keys of size $\Theta(n)$ in $O(1)$ rounds, in this section we show that with some structural assumption on the keys and the order one can do much better.
	In particular, we show that if our keys are strings and we consider the lexicographic order, we can always sort them in $O(1)$ rounds, even if there are strings of length $\omega(n)$.

	We introduce a string sorting algorithm that uses the algorithm of~\cref{thm:sorting_esp_g0} as a black box (specifically, our algorithm uses the algorithm for the special case of $\eps=2/3$, which is proved in \cref{lem:sorting_objects}), and apply a technique called \emph{renaming} to reduce long strings into shorter strings.
	The reduction preserves the original lexicographic order of the strings.
	After applying the reduction several times, all the strings are reduced to strings of length $O(n^{1/3})$ which are sorted by an additional call to the algoirthm of \cref{thm:sorting_esp_g0}.

	\para{Renaming.}
	The idea behind the renaming technique is that to compare two long strings, one can partition the strings into blocks of the same length, and then compare pairs of blocks in corresponding positions in the two strings.
	The lexicographic order of the two original strings is determined by the lexicographic order of the first pair of blocks that are not the same.
	Such a comparison can be done by replacing each block with the block's rank among all blocks, transforming the original strings into new, shorter strings.

	As a first step, the algorithm splits each string into blocks of length $\ceil{n^{1/3}}$ characters\footnote{For the sake of simplicity, we assume from now on that $n^{1/3}$ is an integer and simply write $n^{1/3}$ instead of $\ceil{n^{1/3}}$. } as follows.
	A string $S$ of length $|S|=\ell$ is partitioned into $\ceil{ \frac \ell{n^{1/3}}}$ blocks, each of length $n^{1/3}$ characters except for the last block which is of length $\ell \modulo n^{1/3}$ characters (unless $\ell$ is a multiple of $n^{1/3}$, in which case the length of the last block is also $n^{1/3}$ characters).

	In the case that every string is stored in one node, such a partitioning can be done locally.
	In the more general case, each node $v$ broadcasts the number of strings starting in $v$ and the number of characters $v$ stores from $v$'s last string.
	This broadcast is done in $O(1)$ rounds by Lenzen's routing scheme and using this information
	each node $v$ computes the partitioning positions of the strings stored in $v$.
	Finally, a block that is spread among two (or more) nodes is transferred to be stored only in the node where the block begins.
	This transfer of block parts is executed in $O(1)$ rounds since each node is the source and destination of at most $n^{1/3}$ characters.

	The next step of the algorithm is to sort all the blocks, using the algorithm of \cref{thm:sorting_esp_g0}.
	The result of the sorting is the rank of every block.
	In particular, if the same block appears more than once, all the block's occurrences will have the same rank, and for every two different blocks, the order of their rank will match their lexicographic order.
	Thus, the algorithm will define new strings by replacing each block with its rank in the sorting.
	As a result, every string of length $\ell$ will be reduced to length $\ceil{\frac {\ell}{n^{1/3}}}$. Notice that the alphabet of the new strings is the set of ranks, which is a subset of $[1..n^2]$, and therefore each new character uses $O(\log n)$ bits.

	In the following lemma we prove that the new strings preserve the lexicographic order of the original strings.
	\begin{lemma}\label{lem:renaming_works}
		Let $A$ and $B$ be two strings, and let $A', B'$ be the resulting strings defined by replacing each block of $A$ and $B$ with the block's rank among all blocks, respectively.
		Then,  $A\preceq B$ if and only if $A'\preceq B'$.
	\end{lemma}
\begin{proof}
	We first prove that $A\preceq B \Rightarrow A'\preceq B'$.
	Recall that by \cref{def:lex_order} there are two options for $A\preceq B$ to be hold:
	\begin{enumerate}
		\item If $\ell=\LCP(A,B)<\min\{|A|,|B|\}$ and $A[\ell+1]<B[\ell+1]$
		\item $A$ is a prefix of $B$, i.e. $|A|\le |B|$ and $A=B[1..|A|]$.
	\end{enumerate}

	For any $j$ let $A_j$ and $B_j$ be the $j$th blocks of $A$ and $B$, respectively.

	For the first case, let $\alpha=\ceil{\frac {\ell+1}{n^{1/3}}}$.
	Notice that $A[\ell+1]$ and $B[\ell+1]$ are contained in $A_\alpha$ and $B_\alpha$, respectively.
	Thus, for all $j<\alpha$ we have $A_j=B_j$
	and for $j=\alpha$ we have $A_\alpha\ne B_\alpha$.
	Moreover, by definition $A_\alpha\preceq B_\alpha$.
	Hence, $A'[1..\alpha-1]=B '[1..\alpha-1]$ and $A'[\alpha]<B'[\alpha]$ which means that $A'\preceq B'$.

	For the second case, where $A$ is a prefix of $B$ there are three subcases:
	(1) If $A=B$ then all the blocks in $A$ will get exactly the same ranks as the corresponding blocks in $B$ and therefore $A'=B'$.
	(2) Otherwise, $|A|<|B|$.
	(2a) If $|A|$ is a multiple of $n^{1/3}$ then all the blocks of $A$ are exactly the same as the corresponding blocks of $B$, but $B$ has at least one additional block.
	Therefore, $A'$ is a prefix of $B'$ and $A'\preceq B'$.
	(2b) If $|A|$ is not a multiple of $n^{1/3}$ then the last block of $A$ is shorter than the corresponding block of $B$.
	Let $\alpha=\ceil{\frac{|A|}{n^{1/3}}}$ be the index of the last block of $A$, in particular the block  $A_\alpha$ is a prefix of the block $B_\alpha$, and therefore by definition the rank of $A_\alpha$ is smaller than the rank of $B_\alpha$.
	Thus, we have that $A'[1..\alpha-1]= B'[1..\alpha-1]$ and $A'[\alpha]<B'[\alpha]$ which means that $A'\preceq B'$.

	By very similar arguments, one can prove that $B\prec A \Rightarrow B'\prec A'$. Thus, the claim  $A'\preceq B' \Rightarrow A\preceq B$ follows.
\end{proof}

	\begin{algorithm}[t]		\caption{String sorting}\label{alg:n2}
		\KwIn{A general set of strings, whose total length is $O(n^2)$ characters.}
		\SetKwFor{RepTimes}{Repeat}{times:}{end}

		\RepTimes{7}{

			Each node $v$ broadcasts the number of strings starting in $v$ and the number of characters $v$ stores from $v$'s last string\label{algn2:ln1}\;

			Each node $v$ computes for each string $S$ in $v$, positions of $S$ that are a multiple of $n^{1/3}$. These positions define the borders of $S$'s blocks \label{algn2:ln2}\;

			Each node that stores the beginning of a block $b$ that ends in other nodes collects the rest of $b$ from the succeeding nodes\label{algn2:ln3}\;

			Sort all the blocks, using \cref{alg:n23}\label{algn2:ln4}\;

			Each node replaces each block with its rank in the sorting\label{algn2:ln5}\;

		}

		Sort all the new strings, using \cref{alg:n23}\label{algn2:ln7}\;

	\end{algorithm}

	We are now ready to prove \cref{thm:string_sorting} (see also \cref{alg:n2}).
	\begin{proof}[{Proof of \cref{thm:string_sorting}}]
		The algorithm repeats the renaming process  $7$ times.
		Since the original longest string is of length $O(n^2)$, after seven iterations, we get that the length of every string is $\ceil{O({\frac{n^2}{(n^{1/3})^7}})}=1$.
		Thus, at this time, each string is one block of length $1=O(n^{1/3})$ characters.
		In particular, each string is stored in one node.
		The algorithm uses one more time the algorithm of \cref{thm:sorting_esp_g0}
        (one could also use the sorting algorithm of \cite[{Corollary 4.6}]{Lenzen13}) to solve the problem.
		Notice that during the execution of the renaming process, each block is moved completely to the first node that holds characters from this block.
		In particular, at the final sorting, each string contains one block, that starts at the original node where the complete string starts. Therefore, the ranks found by the sorting will be stored in the required nodes.
	\end{proof}

	\section{Pattern Matching}\label{sec:pm}
	In this section we prove \cref{thm:pm} and introduce an $O(1)$-round \cc algorithm for the pattern matching problem.
	Recall that the input for this problem is a pattern string $P$, and a text string $T$.
	Moreover, we assume that $|P|+|T|=O(n^2)$, since for larger input one cannot hope for an $O(1)$ rounds algorithm, due to communication bottlenecks.
	The goal is to find all the offsets $i$ such that $P$ occurs in $T$ at offset $i$.
	Formally, for every two strings $X,Y$ let $\PM(X,Y)=\{i \mid Y[i+1..i+|X|]=X\}$ be the set of occurrences of $X$ in $Y$.
	Our goal is to compute $\PM(P,T)$ in a distributed manner.
	We introduce an algorithm that solves this problem in $O(1)$ rounds.
	Our algorithm distinguishes between the case where $|P|\le n$ and $|P|> n$.
	Therefore, as a first step, each node $v$ broadcasts the number of characters from $P$ in $v$ and then computes $|P|$.
	In~\cref{sec:shortpattern} we take care of the simple case where $|P|\le n$ and in \cref{sec:longpattern} we describe an algorithm for the case of $n< |P|\in O(n^2)$.

	\subsection{Short Pattern}\label{sec:shortpattern}
	If $|P|\le n$ then the algorithm first broadcasts $P$ to all the nodes using \cref{thm:gen_routing}.
	In addition, we want that every substring of $T$ of length $|P|$ will be stored completely in (at least) one node.
	For this, each node announces the number of characters from $T$ it holds.
	Then, each node that its last input character is the $i$th character of $T$, needs to get the values of $T[i+1],T[i+2],\dots, T[i+|P|-1]$ from the following nodes.
	So, each node that gets $T[j]$ sends it to all preceding nodes starting from the node that gets $T[j-|P|+1]$ (or the node that gets $T[1]$ if $j-|P|+1<1$).
	Notice that all these nodes form a range and that each node will get at most $|P|\le n$ messages.
	Thus, this routing can be done in $O(1)$ rounds using \cref{thm:gen_routing}.
	Now, every substring of $T$ of length $|P|$ is stored completely in one of the nodes, and all the nodes have $P$.
	Thus, every node locally finds all the occurrences of $P$ in its portion of $T$.
 To conclude we proved \cref{thm:pm} for the special case where $|P|\le n$ (see also \cref{alg:shortpm}).

 \begin{algorithm}[t]
	\caption{Pattern matching with short pattern}\label{alg:shortpm}
	\KwIn{Two strings $P$ and $T$, such that $|P|+|T|=O(n^2)$ and $|P|\le n$.}
	Each node $v$ broadcasts the number of characters in $v$\label{algspm:ln0}\;

	Each node that its last input character is the $i$th character of $T$, collects the values of $T[i+1],T[i+2],\dots, T[i+|P|-1]$ from the succeeding nodes\label{algspm:ln1}\;

	$P$ is broadcast to all the nodes\label{algspm:ln2}\;

	Each node locally finds all the occurrences of $P$ in its portion of $T$\label{algspm:ln3}\;

\end{algorithm}

 \begin{algorithm}[b]
	\caption{Pattern matching with $|P|>n$}\label{alg:pm}
	\KwIn{Two strings $P$ and $T$, such that $|P|+|T|=O(n^2)$ and $|P|> n$.}

	Let $B=P[1..n]$ and $E=P[|P|-n\dots |P|]$.
	The algorithm uses~\cref{alg:shortpm} to find all occurrences of $B$ and $E$ in $P$ and $T$\label{algpm:ln1}\;

	Each node encodes the offsets of all the occurrences of $B$ and $E$ it holds in $O(1)$ words of space, and announces these offsets to all other nodes\label{algpm:ln2}\;

	The strings of $\S$ are defined as in Equations \ref{eq:define_SP}, \ref{eq:define_ST} and \ref{eq:define_S} based on the offsets of $B$ and $E$ in $P$ and $T$.  Sort all strings of $\S$ using~\cref{alg:n2}\label{algpm:ln3}\;

	Each node that gets the rank of strings of $\S$ broadcasts this rank\label{algpm:ln4}\;

	Each node computes locally whether $T[i+1..i+|P|]=P$  by using the conditions of \cref{lem:verify_occurrence}, for every $i$ such that $T[i+1]$ stored in the node. Checking the third condition for indices where the first two conditions holds is done by comparing ranks of strings in $\S$ (see \cref{lem:cor_blocks_in_s})\label{algpm:ln5}\;

    \end{algorithm}

	\subsection{Long Pattern}\label{sec:longpattern}

	From now on we consider the case where $|P|> n$.
	To conclude that an offset $i$ of the text is an occurrence of $P$, i.e. that $T[i+1..i+|P|]=P$, we have to get for any $1\le j\le |P|$ evidence that $T[i+j]=P[j]$.
	We will use two types of such evidence.
	The first type is finding all the occurrences of the $n$-length prefix and suffix of $P$ in $P$ and $T$.
	The second type of evidence will be based on the string sorting algorithm of  \cref{thm:string_sorting}, to sort all the blocks between occurrences that were found in the first step, both in $P$ and $T$.
	At every occurrence of $P$ in $T$, all the occurrences of the prefix and suffix will match, and also the blocks between these occurrences will also be the same in the pattern and text (see also \cref{alg:pm}).

	\para{First step - searching for the prefix and suffix.}
	Let $B=P[1..n]$ and $E=P[|P|-n+1..|P|]$ be the prefix and suffix of $P$ of length $n$, respectively.
	We use~\cref{alg:shortpm} four times, to find all occurrences of $B$ and $E$ in $P$ and $T$ (in every execution the algorithm ignores the parts of the original input which are not relevant for this execution).
	By the following lemma, all the locations of $B$ and $E$ found by one node can be stored in $O(1)$ words of space (which are $O(\log n)$ bits). The lemma is derived from \cref{lem:BG_progression} by using the so-called standard trick.
	\begin{lemma}\label{lem:standard_trick}
		Let $X$ and $Y$ be two strings such that $|Y|\ge|X|$ and $|Y|=O(|X|)$.
		Then $\PM(X,Y)$ can be stored in $O(1)$ words of space.
	\end{lemma}

 \begin{proof}
We divide $Y$ into parts of length $2|X|-1$ characters, with overlap, such that each substring of length $|X|$ is contained in one part.
For any $i=0,1,\dots \floor{\frac {|Y|}{|X|}}-1$ let $Y_i=Y[i\cdot|X|+1..\min\{(i+2)\cdot |X|-1,|Y|\}]$ be the $i$th part of $Y$, and let $L_i=\PM(X,Y_i)+i\cdot |X|$ be the set of all the occurrences of $X$ in $Y_i$.
Notice that every occurrence $t$ of $X$ in $Y$ is an occurrence of $X$ in one part of $Y$, specifically in $Y_{\floor{{t}/{|X|}}}$.
Thus, to represent $\PM(X,Y)$ it is enough to store all the occurrences of $X$ in every part $Y_i$.
Since we consider just $\frac {|Y|}{|X|}=O(1)$ parts of $Y$, it is enough to show that $L_i$ can be stored in $O(1)$ words of space.
If $|L_i|\le 2$ then of course it could be stored in $O(1)$ words of space.
Otherwise, $|L_i|\ge 3$, and notice that for every two occurrences in $L_i$ their distance is at most $|Y_i|-|X|+1=|X|$.
Thus, by \cref{lem:BG_progression}, $L_i$ forms an arithmetic progression and therefore can be represented in $O(1)$ words of space by storing only the first and last element and the difference between elements.
\end{proof}

	Hence, every node broadcasts to the whole network all the locations of $B$ and $E$ in the node's fragments of $P$ or $T$.
	Combining all the broadcasted information, each node will have $\PM(B,P)$, $\PM(B,T)$, $\PM(E,P)$ and $\PM(E,T)$.

	\para{Second step - completing the gaps.}
	In the second step, we want to find evidence of equality for all the locations in $P$ and $T$ which do not belong to any occurrence of $B$ or $E$.
	Notice that $|B|=|E|=n$ and therefore every occurrence of $B$ or $E$ that starts at offset $i$ covers all the range $[i+1,i+n]$.
	We will use all the occurrences of $B$ and $E$ that were found in the first step, to focus on the remaining regions which are not covered yet.
	Formally, we define the sets of uncovered locations in $P$ and $T$ as follows. For a string $X$ let
	\begin{gather*}
 	R_X=[1..|X|]\setminus\bigcup_{i\in\PM(B,X)\cup\PM(E,X)}[i+1..i+n]
	\end{gather*}
	The algorithm uses $R_P$ and $R_T$ to define a (multi)set of strings which contains all the maximal regions of $P$ and $T$ which were not covered on the first step (see~\cref{fig:textpatternpm}):
	\begin{gather}
	\S_P=\{P[i..j]| [i..j]\subseteq R_P \text{ and } i-1\notin R_P \text { and }j+1\notin R_P\} \label{eq:define_SP}\\
	\S_T= \, \{T[i..j]| [i..j]\subseteq R_T \text{ and } i-1\notin R_T \text { and }j+1\notin R_T\}.\label{eq:define_ST} \\
	\S=\S_P\cup \S_T\label{eq:define_S}
	\end{gather}

 \begin{figure}[h!]
		\begin{center}
			\includegraphics[width=0.9\textwidth]{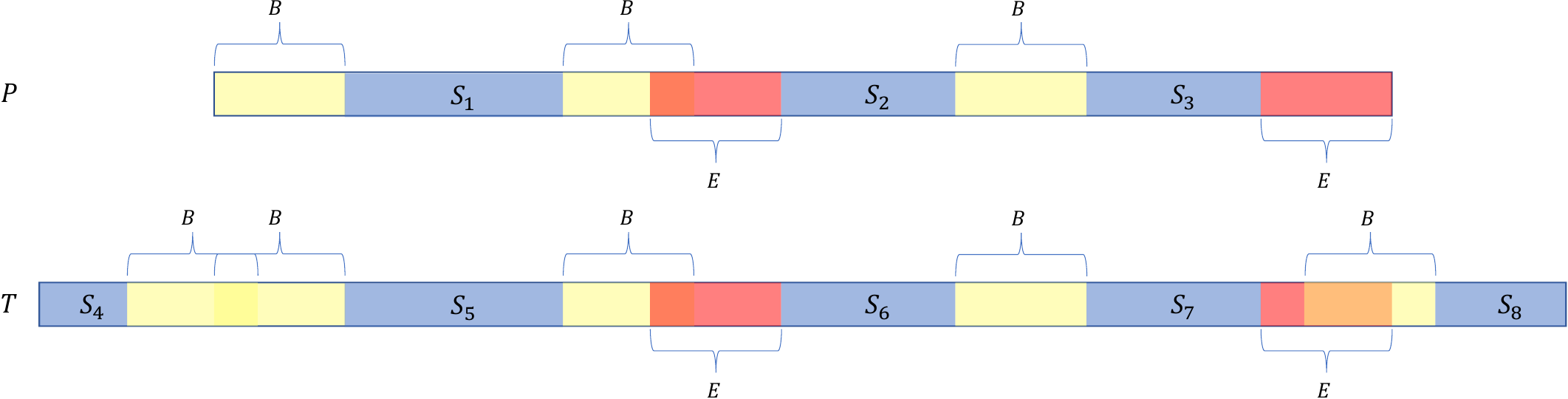}
		\end{center}
		\caption{The yellow regions and red regions represents occurrences of $B$ and $E$, respectively.
			The blue regions of the pattern and the text represent $R_P$ and $R_T$, respectively.
			Notice that $\S_P=\{S_1,S_2,S_3\}$, and $\S_T= \{S_4,S_5,S_6,S_7\}$.
			Given that the pattern is aligned to the $i$th offset, then $i\in PM(P,T)$ if and only if $S_1=S_5$, $S_2=S_6$ and $S_3=S_7$.}
		\label{fig:textpatternpm}
	\end{figure}

	Notice that the total size of all the strings in $\S$ is $O(n^2)$, since $|P|+|T|=O(n^2)$.
	The algorithm uses the string sorting algorithm of \cref{thm:string_sorting}, to sort all the strings in $\S$.
	As a result, every string is stored with its rank - which is the same for two identical strings.
	A useful property of $\S$ is that it contains $O(n)$ strings, as we prove in the next lemma.

	\begin{lemma}
		$|\S|=O(n)$.
	\end{lemma}
	\begin{proof}
		We start with bounding $|\S_T|$.
		By definition of $\S_T$, for $T[i..j]\in \S_T$ we have $i-1\notin R_T$.
		Moreover, by definition of $R_T$, it must be the case that $[i-n..i-1]\cap R_T=\emptyset$.
		Thus, one can associate with every string in $\S_T$ (except for the first string in $\S_T$ if it is a prefix of $T$) a set of $n$ unique locations from $[1..|T|]$ that are not in $R_P$.
		Therefore, $n\cdot (|\S_T|-1)\le|[1..|T|]|=O(n^2)$ and so $|\S_T|=O(n)$.
		Applying the same argument for $\S_P$ gives us $|\S_P|=O(n)$ and $|\S|= O(n)+O(n)=O(n)$.
	\end{proof}

	After sorting the strings of $\S$, each node $v$ broadcasts the ranks of the strings starting at $v$.
	Since there are just $O(n)$ strings in $\S$, this can be done in $O(1)$ rounds with \cref{thm:gen_routing}.
	Therefore, at the end of the second phase, every node has all the
	ranks of strings in $\S$.
	Recall that at the end of the first phase each node stores  $\PM(B,P)$, $\PM(B,T)$, $\PM(E,P)$ and $\PM(E,T)$.
	Hence, for every offset $i$, to check whether $T[i+1..i+|P|]=P$, each node first verifies that all the occurrences of $B$ and $E$ in $P$ match the corresponding occurrences of $B$ and $E$ in
 $T[i+1..i+|P|]$.
	If this is the case, then it must be that all the maximal regions in $R_P$ are in \emph{corresponding positions} to the maximal regions in $R_T$ at $[i+1..i+|P|]$.
	Thus, the node compares the rank of every string in $\S_P$ with the rank of the corresponding (due to shift $i$) string of $\S_T$ and checks whether they are equal (see also \cref{fig:textpatternpm} and \cref{alg:pm}).

	The following two lemmas give us the mathematical justification for the last part of the algorithm.
\cref{lem:verify_occurrence} states that the test made by the algorithm is sufficient to decide whether $i\in\PM(P,T)$.
\cref{lem:cor_blocks_in_s} proves that for any $i$ where the first two conditions of \cref{lem:verify_occurrence} holds, the test of the third condition can be made by comparing the ranks of strings from $\S$, just like the algorithm acts.

\begin{lemma}\label{lem:verify_occurrence}
	$T[i+1..i+|P|]=P$ if and only if all the following holds:
	\begin{enumerate}
		\item $\PM(B,P)=(\PM(B,T)-i) \cap [0..|P|-n]$
		\item $\PM(E,P)=(\PM(E,T)-i) \cap [0..|P|-n]$
		\item For every maximal range $[a..b]\sub R_P$ we have $P[a..b]=T[i+a..i+b]$.
	\end{enumerate}
\end{lemma}

\begin{proof}
	The first direction is simple.
	Assume $T[i+1..i+|P|]=P$, then for every $0\le j<|P|-n$ we have  $j\in\PM(B,P)$ if and only if $B=P[j+1..j+n]=T[i+j+1..i+j+n]$  if and only if $i+j\in\PM(B,T)\iff j\in\PM(B,T)-i$.
	A similar argument works for the second property.
	Lastly, $P[a..b]=T[i+a..i+b]$ is a direct consequence of the fact that $T[i+1..i+|P|]=P$.

	For the other direction, let $j\in [1..|P|]$, our goal is to prove that $T[i+j]=P[j]$.
	If $[j-n..j-1]\cap \PM(B,P)\ne \emptyset$ then there exists some $t\in [j-n..j-1]$  such that $t\in\PM(B,P)$.
	By the first property we have $t+i\in\PM(B,T)$. Thus, $P[t+1..t+n]=T[i+t+1..i+t+n]$ and in particular $P[j]=P[t+(j-t)]=T[t+i+(j-t)]=T[i+j]$.
	The case where  $[j-n..j-1]\cap \PM(E,P)\ne \emptyset$ is symmetric.
	The last case we have to consider is where  $[j-n..j-1]\cap \PM(B,P)= \emptyset$ and $[j-n..j-1]\cap \PM(E,P)= \emptyset$.
	In this case, let $P[a..b]$ be the maximal region in $R_P$ that contains $j$ (such a region must exists). By the third property we have $P[a..b]=T[i+a..i+b]$ and in particular $P[j]=P[a+j-a]=T[i+a+j-a]=T[i+j]$ as required.
\end{proof}

	\begin{lemma}\label{lem:cor_blocks_in_s}
		Let $i\in[0..|T|-|P|]$ such that
		$\PM(B,P)=(\PM(B,T)-i) \cap [0..|P|-n]$ and $\PM(E,P)=(\PM(E,T)-i) \cap [0..|P|-n]$.
		Then, for every maximal range $[a..b]\subseteq R_P$  we have $P[a..b]\in \S_P$ and $T[i+a..i+b]\in \S_T$.
	\end{lemma}

\begin{proof}
	First, by definition of $S_P$ we have $P[a..b]\in \S_P$.
	Similarly, to prove that $T[i+a..i+b]\in \S_T$ it is sufficient to prove that $[i+a..i+b]$ is a maximal range in $R_T$ i.e. that  (1) $[i+a..i+b]\subseteq R_T$ and (2) $i+a-1\notin R_T$ and (3) $i+b\notin R_T$.

	For (1), let $j\in[i+a..i+b]$ assume by a way of contradiction that $j\notin R_T$. Then by definition there exists some $t\in[j-n..j-1]$ such that $t\in \PM(B,T),\PM(E,T)$.
	But then, it must be the case that for $t'=t-i$ we have $t'\in \PM(B,P)\cup\PM(E,P)$.
	Hence, and therefore $j-i=t-i+j-t=t'+(j-t)\notin R_P$ but $j-i\in[a..b]$ and therefore $[a..b]\not\subseteq R_P$ in contradiction.
	Therefore, $[i+a..i+b]\subseteq R_T$.

	For (2), since $[a..b]$ is a maximal range in $R_P$ we have $[a-1]\notin R_P$. Moreover, since $0\in\PM(B,P)$ it must be that $a>n$ and therefore by definition of $R_P$ we have $a-1-n\in\PM(B,P)$.
	Hence, $a-1-n+i\in\PM(B,T)$ and therefore $a-1-n+i+n=a-1+i\notin R_T$.

	Similarly for (3), since $[a..b]$ is a maximal range in $R_P$ we have $[b+1]\notin R_P$. Moreover, since $|P|-n\in\PM(E,P)$ it must be that $b<|P|-n$ and therefore by definition of $R_P$ we have $b\in\PM(B,E)$.
	Hence, $i+b\in\PM(B,T)$ and therefore $i+b\notin R_T$.
	Thus, we proved that $[i+a..i+b]$ is a maximal range in $R_T$ and therefore $T[i+a..i+b]\in\S_T$.
\end{proof}

\section{Suffix Array Construction and the Corresponding LCP Array}\label{sec:SA_and_LCP}
	In this section, we are proving \cref{thm:SA_LCP} by introducing an algorithm that computes $\SA_S$, the suffix array of a given string $S$ of length $O(n^2)$ in the \cc model in $O(\log\log n)$ rounds.
	Moreover, we show how to compute the complementary $\LCP_S$ array in the same asymptotic running time.
 We first give the formal definition  of $\SA_S$ and $\LCP_S$:
    \begin{definition}
		For a string $S$, the suffix array, denoted by $\SA_S$ is the sorted array of $S$ suffixes, i.e., $S[\SA_S[i]..]\prec S[\SA_S[i+1]..]$ for all $1\le i<|S|-1$.
		The corresponding $\LCP$ array, $\LCP_S$, stores the $\LCP$ of every two consecutive suffixes due to the lexicographical order. Formally $\LCP[i]=\LCP(S[\SA_S[i]..],S[\SA_S[i+1]..])$.
    \end{definition}

	Our algorithm follows the recursive process described by  Pace and Tiskin~\cite{PT13}, which is a speedup of  K{\"{a}}rkk{\"{a}}inen et al.~\cite{KSB06} for parallel models.
	The main idea of the recursion is that in every level, the algorithm creates a smaller string that represents a subset of the original string positions, solve recursively and use the results of the subset to compute the complete results.
	While the depth of the recursion increases, the ratio between the length of the current string and the length of the new string increases as well.
	At the beginning the ratio is constant and in $O(\log\log n)$ rounds it becomes polynomial in $n$.
	The main difference between our algorithm and \cite{PT13,KSB06} is that our algorithm is simpler due to the powerful sorting algorithm provided in \cref{thm:sorting_esp_g0}.
	Moreover, we also introduce how to compute $\LCP_S$, in addition to $\SA_S$.
	We first ignore the $\LCP_S$ computation and then in \cref{sec:LCP} we give the details needed for computing $\LCP_S$.

    	Our algorithm uses the notion of \emph{difference cover}~\cite{CL00}, and \emph{difference cover sample} as defined by K{\"{a}}rkk{\"{a}}inen et al.~\cite{KSB06}.
	For a parameter $t$, a difference cover $DC_t\subseteq [0,t-1]$ is a set of integers such that for any pair $i,j\in\mathbb Z$ the set $DC_t$ contains $i',j'\in DC_t$ with $j-i\equiv j'-i' (\modulo t)$.
	For every $t\in\mathbb{N}$ there exists a difference cover $DC_t$ of size $\Theta(\sqrt t)$, which can be computed in $O(\sqrt t)$ time (see Colbourn and Ling~\cite{CL00}).
	For a string $S$, \cite{KSB06} defined the difference cover sample $DC_t(S)=\{i\mid i\in[1..|S|]\text{ and } (i\modulo t)\in DC_t\}$, as the set of all indices in $S$ which are in $DC_t$ modulo $t$.
	The following lemma was proved in \cite{SV13}.
	\begin{lemma}[{\cite[{Lemma 2}]{SV13}}]\label{lem:difference_cover}
		For a string $S$ and an integer $t\le |S|$, there exists a difference cover $DC_t$, such that $|DC_t(S)|\in O(|S|/\sqrt t)$ and for any pair of positions $i,j\in[1..|S|]$ there is an integer $k\in[0..t-1]$ such that $(i+k)$ and $(j+k)$ are both in $DC_t(S)$.
	\end{lemma}

	At every level of the recursion let $\eps>0$ be a number such that the length of the string $S$ is $|S|\in O(n^{2-\varepsilon})$.
	Later we will describe how to choose $\eps$ exactly (see the time complexity analysis).
	At the first level $\eps=O(1/\log n)$ satisfies this requirement.
	Let $DC_t\subset[0..t-1]$ be some fixed difference cover with $t=\min\{n^{\eps},n^{1/3}\}$
	of size $|DC_t|=O(\sqrt t)$.

	For every $i\in [1..|S|]$ let  $S_i=S[i..i+t-1]$ be the substring of $S$ of length $t$ starting at position $i$ (we assume that $S[j]$ is some dummy character for every $j\ge|S|$).
	As a first step the algorithm sorts all the strings in $\mathcal S=\{S_i\mid i\in DC_t(S)\}$.
	Notice that the total length of all these strings is $|DC_t(S)|\cdot t=O(\frac{|S|}{\sqrt t})\cdot t= O(|S|\sqrt t)\subseteq O(n^{2-\eps}\cdot n^{\eps/2})=O(n^{2-\varepsilon/2})\subseteq O(n^2)$ and therefore the algorithm can sort all the strings in $O(1)$ rounds using \cref{thm:sorting_esp_g0}.
	Recall that as a result of running the sorting algorithm, the node that stores index $i$ of $S$, has now the rank of $S_i$, $\rank_{\mathcal S}(S_i)$,  among all the strings of $\mathcal S$ (copies of the same string will have the same rank).
	The algorithm uses the ranks of the strings to create a new string of length $|DC_t(S)|$.
	For every $i\in DC_t$ let $S^{(i)}=\rank_{\S}(S_i)\rank_{\S}(S_{i+t})\rank_{\S}(S_{i+2t})\rank_{\S}(S_{i+3t})\dots$ (if $0\in DC_t$ then $S^{(0)}$ starts from $\rank_{\S}(S_t)$ since $S_0$ does not exist).
	Moreover, let $S'$ be the concatenation of all $S^{(i)}$s for $i\in DC_t$ (in some arbitrary order) with a special character $'\$'$ as a delimiter between  the strings $S^{(i)}$.
	The algorithm runs recursively on $S'$.
	The result of the recursive call is the suffix array of $S'$, $\SA_{S'}$ (which stores the order of the suffixes in $S'$).
	Note that every index $i\in DC_t(S)$ has a corresponding index in $S'$ which is where $\rank_{\S}(S_i)$ appears, we denote this position as $f(i)$.
	The algorithm sends for every index $i$ the rank of $f(i)$ (which is the index in $\SA_{S'}$ where $f(i)$ appears) to the node that stores index $i$ of $S$.

	Due to the following claim, which we prove formally later in \cref{lem:LCP_and_order_sync}, the order of suffixes of $S$ from $DC_t(S)$ is the same as the order of the corresponding suffixes of $S'$.
	\begin{claim}\label{lem:recursive_SSA}
	For $a,b\in DC_t(S)$ we have
	 $S[a..]\prec S[b..]$ if and only if $S'[f(a)..]\prec S'[f(b)..]$.
	\end{claim}

	Due to \cref{lem:recursive_SSA}, $\SA_{S'}$ represents the order of the subset of suffixes of $S$ starting at $DC_t(S)$.
	To extend the result for the complete order of all the suffixes of $S$ (hence, computing $\SA_S$), the algorithm creates for every index of $S$ a \emph{representative object} of size $O(t)$ words of space.
	These objects have the property that by comparing two objects one can determine the order of the corresponding suffixes.

	\para{The representative objects.}
	For every index $i$ the object represents the suffix $S[i..]$ is composed of two parts.
	The first part is $S_i$ - which is the substring of $S$ of length $t$ starting at position $i$.
	The second part is the ranks (due to the lexicographic order) of all the suffixes at position in $DC_t(S)\cap[i..i+t-1]$ among all suffixes of $DC_t(S)$, using $\SA^{-1}_{S'}$.
	This information is stored as an array $A_i$ of length $t$ as follows.
	For every $j\in [0..t-1]$ if $i+j\in DC_t(S)$ we set $A_i[j] = \SA^{-1}_{S'}[f(i+j)]$, which is the rank of position $i+j$ (as computed by the recursive call) and $A_i[j]=-1$ otherwise.
	The first part is used to determine the order of two suffixes which their $\LCP$ is at most $t$ and the second part is used to determine the order of two suffixes which their $\LCP$ is larger than $t$.

	The comparison of the objects represent positions $a$ and $b$, is done as follows.
	The algorithm first compares $S_a$ and $S_b$.
	If $S_a\ne S_b$ then the order of $S[a..]$ and $S[b..]$ is determined by the order of $S_a$ and $S_b$.
	Otherwise,  $S_a=S_b$ and the algorithm uses the second part of the objects.
	By \cref{lem:difference_cover} there exists some $k\in[0..t-1]$ such that $a+k$ and $b+k$ are both in $DC_t(S)$.
	In particular $A_a[k]$ and $A_b[k]$ both hold actual ranks of suffixes of $S'$ (and not $-1$s).
 	The algorithm uses $A_a[k]$ and $A_b[k]$ to determine the order of $S[a..]$ and $S[b..]$.
 	By the following lemma the order of $A_a[x]$ and $A_b[x]$ is exactly the same order of the corresponding suffixes starting at positions $a$ and $b$.

 	\begin{lemma}\label{lem:rep_objects}
 		Let $a,b\in[|S|]$ such that $S[a..]\prec S[b..]$ and $S_a=S_b$.
 		Then, for every $k$, if $A_a[k]\ne -1$ and $A_b[k]\ne-1$ then $A_a[k]<A_b[k]$.
 		Moreover, there exists some $k\in[0..t-1]$ such that  $A_a[k]\ne -1$ and $A _b[k]\ne-1$.
 	\end{lemma}
	 \begin{proof}
	 	Let $\ell=\LCP(S[a..],S[b..])$.
	 	By definition, $S[a..a+\ell-1]=S[b..b+\ell-1]$ and $S[a+\ell]<S[b+\ell]$.
	 	Notice that $\ell\ge t$.
	 	Thus, for every $0\le k\le t\le\ell$ we have $S[a+k..a+k+(\ell-k-1)]=S[b+k..b+k+(\ell-k-1)]$ and $S[a+k+(\ell-k)]<S[b+k+(\ell-k)]$.
	 	Therefore, $S[a+k..]\prec S[b+k..]$.
	 	For $k\in[0..t-1]$ such that $A_a[k]\ne -1$ and $A_b[k]\ne-1$, it must be the case that  $a+k,b+k\in DC_t(S)$.
	 	Thus, by \cref{lem:recursive_SSA} we have $S'[f(a+k)..] \prec S'[f(b+k)..]$ and therefore $A_a[k]<A_b[k]$.

 		By \cref{lem:difference_cover} there exists some $k\in[0..t-1]$ such that $a+k$ and $b+k$ are both in $DC_t(S)$.
 		For this value of $k$ it is guaranteed that  $A_a[k]\ne -1$ and $A_b[k]\ne-1$.
	 \end{proof}

For every index $i\in[1..|S|]$ the algorithm creates an object of size $O(t)$ words of space.
Thus, the total size of all the objects is $O(t|S|)\subseteq O(n^\eps\cdot n^{2-\eps})=O(n^2)$.
Moreover, by definition $t\le n^{1/3}$.
Therefore the algorithm sorts all the objects in $O(1)$ rounds using the algorithm of \cref{thm:sorting_esp_g0}.
By \cref{lem:rep_objects} the result of the object sorting algorithm is a sorting of all the suffixes of $S$.

\para{Time complexity.}
Recall that for $|S|=O(n^{2-\eps})$ we defined $t=\min\{n^{\eps}, n^{1/3}\}$ and that the length of $S'$ is $|S'|=|DC_t(S)|= O(|S|/\sqrt t)$.
Let $c$ be a constant such that $|S'|\le c|S|/\sqrt t$ (for any $n>n_0$ for some $n_0$).
Denote $S_k$, $S'_k$, $\eps_k$ and $t_k$ as the values of $S$, $S'$, $\eps$ and $t$ in the $k$th level of the recursion, respectively.
Our goal is to gurantee an exponantial growth in the value of $\eps$ from level to level.
For the first level, in order to have a growth, we have $|S'_1|\le\frac{c|S_1|}{\sqrt{t_1}}=\frac{c|S_1|}{n^{0.5\eps_1}}=\frac {c}{n^{0.1\eps_1}}\frac{|S_1|}{n^{0.4\eps_1}}=\frac {c}{n^{0.1\eps_1}}O(n^{2-1.4\eps_1}).$
We are setting $\eps_1=10\log c/\log n$ and so $|S'_1|\le O(n^{2-1.4\eps_1})$.
Then, as long as $\eps_i<1/3$ we define $\eps_{i+1}=1.4\eps_i$ and we get with similar analysis $|S'_i|\le O(n^{2-1.4\eps_i})$.
By a straightforward induction we get $|S'_i|\le O(n^{2-{1.4}^i\eps_1})$.
From the time  that $\eps_i\ge 1/3$ (i.e. $|S_i|=O(n^{5/3})$)  we have $t_{i+1}=n^{1/3}$.
Thus, $|S'_{i+1}|=O(|S_i|/\sqrt t)=O(|S'_i|/n^{1/6})= O(n^{3/2})$ and in four rounds we get $S'_{i+4}=O(n)$.
At this time the algorithm solves the problem in $O(1)$ rounds in one node.
Therefore the algorithm performs  $O(\log\log n)+4=O(\log\log n)$ levels of recursion.
Each level of recursion takes $O(1)$ rounds, and therefore the total running time of the algorithm is $O(\log\log n)$.

\begin{remark*}
	As described in \cref{sec:intro}, our algorithm can be translated to $O(\log\log n)$ rounds algorithm in the MPC model.
	However, if one consider a variant of the MPC model where the product of machines and memory size is polynomialy larger than $n$, i.e. if $M\cdot S= n^{1+\alpha}$ for some constant $\alpha>0$, then there is a faster algorithm.
	In particular one can use $t=n^{\alpha}$ in all rounds and get an $O(1/\alpha)=O(1)$ rounds algorithm.
 	\end{remark*}

\subsection{Computing the Corresponding LCP Array}\label{sec:LCP}
	The computation of $\LCP_S$ is done during the computation of $\SA_S$, by several additional operations at some steps of the computation.
	The recursive process has exactly the same structure, but now every level of the recursion can use both $\SA_{S'}$ and $\LCP_{S'}$ and has to compute $\LCP_{S}$ in addition to $\SA_S$.

	Recall that in the $\SA$ construction algorithm of the previous section, the order of two suffixes of $S$ that starts in $DC_t(S)$ is exactly the same as the order of the corresponding suffixes of $S'$ that is represented in $\SA_{S'}$.
	The issue with computing $\LCP_S$ is slightly more complicated.
	In the following lemma we prove that for two positions $a,b\in DC_t(S)$ we have
	$\LCP(S'[f(a)..],S'[f(b)..])= \floor{\LCP(S[a..],S[b..])/t}$, which means $\LCP(S[a..],S[b..]])\in t\cdot \LCP(S'[f(a)..],S'[f(b)..]) + [0..t-1]$.

	\begin{lemma}\label{lem:LCP_and_order_sync}
		For $a,b\in DC_t(S)$ we have.
		$\LCP(S'[f(a)..],S'[f(b)..])= \floor{\LCP(S[a..],S[b..])/t}$
		and if $S[a..]\prec S[b..]$ then $S'[f(a)..]\prec S'[f(b)..]$.
	\end{lemma}
	\begin{proof}
		Let $\ell=\LCP(S[a..],S[b..])$.
		By definition $S[a..a+\ell-1]=S[b..b+\ell-1]$ and $S[a+\ell]\ne S[b+\ell]$.
		Our goal is to prove that for any $0\le i<\floor{\ell/t}$, $S'[f(a)+i]=S'[f(b)+i]$ and that $S'[f(a)+\floor{\ell/t}]\ne S'[f(b)+\floor{\ell/t}]$.

		Recall that by the definition of $S'$, $S'[f(a)+i]$ is exactly $\rank_{\S}(S_{a+t\cdot i})$  (as long as $i$ is small enough).
		Similarly, $S'[f(b)+i]=\rank_{\S}(S_{b+t\cdot i})$.
		Thus, for any $0\le i<\floor{\ell/t}$ we have
		\begin{align}\label{eq:SA_equality}
		S'[f(a)+i]&=\rank_{\S}(S_{a+t\cdot i})=\rank_{\S}(S[a+ti..a+ti+(t-1)])\\&=\rank_{\S}(S[b+ti..b+ti+(t-1)])=\rank_{\S}(S_{b+t\cdot i})=S'[f(b)+i] \nonumber
		\end{align}
		With similar analysis one can get  $S'[f(a)+\floor{\ell/t}]\ne S'[f(b)+\floor{\ell/t}]$.

		The assumption $S[a..]\prec S[b..]$ means that $S[a+\ell]<S[b+\ell]$.
		Let $k=t\cdot\floor{\ell/t}$, we will focus on $S_{a+k}$ and $S_{b+k}$.
		Let $r=\ell-k$ and notice that $r<t$ (and $r=\ell\modulo t$).
		We will show that $S_{a+\floor{\ell/t}}\prec S_{b+\floor{\ell/t}}$.
		For any $i<r$ since $k+i<k+r=\ell$ we have $S_{a+\floor{\ell/t}}[i]=S[a+k+i]= S[b+k+i]=S_{b+\floor{\ell/t}}[i]$.
		Since $k+r=\ell$ we have $S_{a+\floor{\ell/t}}[r]=S[a+k+r]=S[a+\ell]< S[b+\ell]=S[b+k+r]=S_{b+\floor{\ell/t}}[r]$.
		Therefore, $S_{a+\floor{\ell/t}}\prec S_{b+\floor{\ell/t}}$ and $S'[f(a)+\floor{\ell/t}]=\rank_{\S}(S_{a+\floor{\ell/t}})<\rank_{\S}(S_{b+\floor{\ell/t}})=S'[f(b)+\floor{\ell/t}]$.
		Combining with Equation \ref{eq:SA_equality} we have $S'[f(a)..]\prec S'[f(b)..]$.
	\end{proof}

	\para{First step - compute exact $\LCP$ for $DC_t(S)$.}
	Let $i_1,i_2,\dots,i_{|DC_t(S)|}$ be the indices of $DC_t(S)$ such that for any $j$ we have $S[i_j..]\prec S[i_{j+1}..]$.
	Notice that $i_j=f^{-1}(\SA_{S'}[j])$.
	The first step of the algorithm is to compute the exact value of $\LCP(S[i_j..], S[i_{j+1}..])$ for any $j$.
	For this step the algorithm uses inverse suffix array $\SA^{-1}_{S'}$.
	Recall that for any $i$, $\SA^{-1}_{S'}[i]$ is the index $j$ such that $SA_{S'}[j]=i$.
	By a simple routing, in $O(1)$ rounds, the algorithm distributes the $\SA^{-1}_{S'}$ information to the \cc nodes, such that the node $v$ that stores the $a$th character of $S$ for $a\in DC_t(S)$ will get $\SA^{-1}_{S'}[f(a)]$.
	Moreover, $v$ will get also the value $b$ such that $f(b)=\SA_{S'}[\SA^{-1}_{S'}[f(a)]+1]$ which is the index of the lexicographic successive  suffix among $DC_T(S)$ suffixes.
	In addition $v$ gets $\ell=\LCP_{S'}[\SA^{-1}_{S'}[f(a)]]$ which is exactly $\ell=\LCP(S'[f(a)..],S'[f(b)..])$.
	Now, $v$ creates $2t=O(n^{\eps})$ queries - to get all the characters of $S[a+\ell\cdot t..a+\ell\cdot t+t-1]$ and
	$S[b+\ell\cdot t..b+\ell\cdot t+t-1]$ (each query is for one character).
	Notice that every node holds at most $O(n^{1-\eps})$ indices of $DC_t(S)$, and therefore every node has at most $O(n^{1-\eps}\cdot t)=O(n)$ queries.
	Thus, using \cref{lem:queries_routing} in $O(1)$ rounds all the queries will be answered.
	Using the answers, every index $i_j$ computes the exact value of $\LCP(S[i_j..],S[i_{j+1}..])$.

	Let $\overline\LCP_{S'}$ the array of all the revised $\LCP$ values of $\LCP_{S'}$ i.e., the value of $\overline\LCP_{S'}[i]$ is the exact $\LCP$ of the suffixes $f^{-1}(\SA_{S'}[i])$  and $f^{-1}(\SA_{S'}[i+1])$ which are the lexicographically $i$th and $i+1$th suffixes among $DC_t(S)$.
	We store $\overline{\LCP}_{S'}$ in a distributed manner in the \cc network.

	\para{Second step - compute $\LCP_S$}
	In the second step, the algorithm computes the $\LCP$ of every suffix of $S$ with its lexicographic successive suffix.
	This computation is done similarly to the second step of the $\SA_S$ construction algorithm.
	In every level, the computation of $\LCP_S$ is done after the computation of $\SA_S$.
	Thus, we can use the order of the suffixes of $S$ as an input for this step.
	For every index $i$, let $\hat i$ be the index of the successive suffix to $S[i..]$.
	The node that stores $S[i]$ gets $\hat i$ in $O(1)$ rounds.
	The algorithm uses the same representative objects used to compute $\SA_S$, to compute $\LCP(S[i..],S[\hat i..])=\LCP_S[\SA^{-1}_S[i]]$.
	Recall that the representative objects of $i$ and $\hat i$ are composed of two parts.
	The first part contains $S_i$ and $S_{\hat i}$ which are the substrings of $S$ of length $t$ starting at positions $i$ and $\hat i$.
	The algorithm first compares $S_i$ and $S_{\hat i}$ and if $S_i\ne S_{\hat i}$ then $\LCP(S[i..],S[\hat i..])=\LCP(S_i,S_{\hat i})$.
	If $S_i=S_{\hat i}$, the algorithm uses the second part of the representative objects.
	Recall that in this part the object of an index $a$ stores an array of length $t$, with the ranks (amongs suffix starting at $DC_t(S)$) of suffixes from $DC_t(S)\cap[a..a+t-1]$ and $-1$s.
	Moreover, by \cref{lem:difference_cover} it is guaranteed that there is some $k\in[0..t-1]$ such that $A_i[k]\ne -1$ and $A_{\hat i}[k]\ne -1$.
	Since $S_i=S_{\hat i}$ and $k<|S_i|$, we have  $\LCP(S[i..],S[\hat i..])=k+\LCP(S[i+k..],S[\hat i+k..])$.
	Thus, we have the ranks of the suffixes $i+k$ and $\hat i+k$ among suffixes of $S$ starting at positions in $DC_t(S)$ (which are $\SA^{-1}_{S'}[f(i+k)]$ and $\SA^{-1}_{S'}[f(\hat i+k)]$).
	Due to the following fact,  $\LCP(S[i+k..],S[\hat i+k..])=\min(\overline{\LCP}_{S'}[j]\mid \SA^{-1}_{S'}[f(i+k)]\le j \le  \SA^{-1}_{S'}[f(\hat i+k)])$.
	(This is because $\overline{\LCP}_{S'}$ is an array of the $\LCP$ values of monotone sequence of strings, and $S[i+k),S[\hat i+k)$ are both elements in the sequence).

	\begin{fact}
		Let  $T_1,T_2,\dots, T_k$ be a sequence of strings such that $T_i\prec T_{i+1}$ and $T_i$ is not a prefix of $T_{i+1}$ for all $i\in[1..k-1]$.
		Then, for any $a<b$ we have $\LCP(T_a,T_b)=\min\{\LCP(T_i,T_{i+1})\mid i\in[a..b-1]\}$.
	\end{fact}
	\begin{proof}
		Let $\ell=\min\{\LCP(T_i,T_{i+1} \mid i\in[a..b-1])\}$ and let $i'\in[a..b-1]$ be an index with $\LCP(T_{i'},T_{i'+1})=\ell$.
		For any $1\le j\le \ell$ we have $T_a[j]=T_{a+1}[j]=T_{a+2}[j]=\cdots=T_{b}[j]$ since for every $i\in[a..b-1]$ we have $\LCP(T_i,T_{i+1})\ge \ell$ and by a straightforward induction.
		On the other hand, 	$T_a[\ell]\le T_{a+1}[\ell]\le T_{a+2}[\ell]\le\cdots\le T_{b}[\ell]$, because $T_i\prec T_{i+1}$ for all $i$.
		Moreover, since $T_{i'}[\ell+1]\ne T_{i'+1}[\ell+1]$, it must be that $T_{i'}[\ell+1]< T_{i'+1}[\ell+1]$ and therefore $T_{a}[\ell+1]< T_{b}[\ell+1]$.
		Thus, $\LCP(T_a,T_b)=\ell$.
	\end{proof}
	Thus, for every $i$, if $\LCP(S[i..],S[\hat i..])$ is not determined by $S_i$ and $S_{\hat i}$, the algorithm has to perform one range minimum query (RMQ) on $\overline\LCP_{S'}$.
	Now, we will describe how to compute all these range minimum queries in $O(1)$ rounds.
	This lemma might be of independent interest.

	\begin{definition}
		Given an array $A$ and two indices $i,j$ such that $1\le i\le j \le |A|$, a Range Minimum Query $\RMQ_A(i,j)$ returns the minimum value $x$ in the range $A[i..j]$.
	\end{definition}

	\begin{lemma}\label{lem:RMQ}
		Let $A$ be an array of $O(n^2)$ numbers (each number of size $O(\log n)$ bits), distributed among $n$ nodes in the \cc model, such that each node holds a subarray of length $O(n)$.
		In addition, every node has $O(n)$ $\RMQ$ queries.
		Then, there is an algorithm that computes for each node the results of all the $\RMQ$ queries in $O(1)$ rounds.
	\end{lemma}

	\begin{proof}

		First, each node broadcasts its subarray length, i.e. how many numbers it contains.
		Second, each node broadcasts the minimum number within the node.

		There are two types of $\RMQ$ queries.
		The first type is where the range of the $\RMQ$ is contained in one specific node, i.e. both $i$ and $j$ of the $\RMQ$ are on the same node.
		The seconde type is where $i$ and $j$ are not in the same node.
		For this case, we separate the $\RMQ$ into three ranges.
		The first range is $i$ to $i'$, where $i'$ is the index of the last number in the node that contains the $i$'th number.
		The third range is $j'$ to $j$, where $j'$ is the index of the first number in the node that contains the $j$'th number.
		The second range is $i'+1$ to $j'-1$, i.e. all the indices in the intermediate nodes (this range might be empty).
		To calculate $\RMQ(i,j)$, it is enough to calculate $\RMQ$ on the three ranges, since $\RMQ(i,j) = \min(\RMQ(i,i'),\RMQ(i'+1,j'-1),\RMQ(j',j))$.
		It is easy to calculate $\RMQ(i'+1,j'-1)$, since on the second step of the algorithm, each node broadcasts its minimum number.
		We are left with two $\RMQ$ queries to two specific resolving nodes.

		To conclude, after the second step, the $O(n)$ $\RMQ$ queries on each node can be calculated with $O(n)$ $\RMQ$ queries to specific resolving nodes.
		Since both an $\RMQ$ to a specific resolving node and the result can be encoded with $O(\log n)$ bits.
		Hence, using \cref{lem:queries_routing}, in $O(1)$ rounds all the $O(n)$ $\RMQ$ queries to specific resolving nodes are resolved.
	\end{proof}

	\para{Complexity.}
	The overhead of computing $\LCP_S$ from $\LCP_{S'}$ is just in constant number of rounds per level of the recursion.
	So, in total the computation of $\SA_S$ and $\LCP_S$ takes $O(\log\log n)$ rounds.

	\bibliography{refs}
\newpage
	\appendix

	\section{Sorting Objects of Size \texorpdfstring{$O(n)$}{O(n)}}\label{sec:sorting_eps_0}

	In this section we prove \cref{thm:sorting_esp_0}.
	Notice that \cref{prob_sorting} with $\eps=0$ means that every key is of size $O(n)$ words of space.

	\subsection{Upper Bound}
In this section we show an algorithm that sorts objects of size $O(n)$ words in $O(\log n)$ rounds.
Our algorithm simulates an execution of a sorting network.
A sorting network with $N$ wires (analog to cells of an input array) sorts comparable objects as follows.
The network has a fixed number of parallel levels, each level is composed of at most $N/2$ comparators.
A comparator compares two input objects and swap their positions if they are out of order.
The number of parallel levels of a sorting network is called the \emph{depth} of the sorting network.

	Ajtai, Komlós, and Szemerédi (AKS)~\cite{AKS83} described a sorting network of depth $O(\log N)$ for every $N\in\N$.
	Notice that in our setting $N\in O(n^2)$.
	We prove the following lemma.

	\begin{lemma}\label{lem:aks_simulation}
		There is an algorithm that solves \cref{prob_sorting} with $\eps=0$ in $O(\log n)$ rounds.
	\end{lemma}
		\begin{proof}

		We show how to simulate an execution of each level of AKS sorting network for $N$ input objects in the \cc model in $O(1)$ rounds.

		First, each node broadcasts the number of objects it stores, and then each node calculates $N$, the number of objects in the network, and produces the AKS sorting network with $N$ wires.
		In addition, each node attaches to every object within the node the global index of the object as a metadata.

		On each level, there are $O(n^2)$ comparators.
		Each node $v_i$ is responsible for the $[(i-1)\ceil{N/n}+1..i\ceil{N/n}]= O(n)$ comparators (if exist).
		To do so, each node with inputs to a comparator under $v_i$'s responsibility, sends for every such input the size and the index of the input to $v_i$.
		This routing takes $O(1)$ rounds.
		Denote the sum of the sizes of the objects (inputs) under $v_i$ responsibility as $S_{v_i}$.
		Then, $v_i$ creates $a_i=\ceil{S_{v_i}/n}$ auxiliary nodes.
		Notice that  the total number of auxiliary nodes is at most \[\sum_{i=1}^n a_i=\sum_{i=1}^n \ceil{S_{v_i}/n}\le \sum_{i=1}^n 1 + S_{v_i}/n\le n+O(n^2)/n=O(n).\]
		By \cref{lem:aux_ndes} the algorithm can simulate each round with the auxiliary nodes on the original network in $O(1)$ rounds.
		Then, $v_i$ calculates a partition of the comparators between $v_i$'s auxiliary nodes such that for each auxiliary node of $v_i$, the total size of objects for $v_i$'s comparators to this auxiliary node is $O(n)$ (by \cref{lem:dist_info}, there is such a partition).
		Then, let $u$ be a node that holds an object $B$ that should be sent to one of $v_i$'s comparators.
		$v_i$ sends to $u$, which auxiliary node is the target of $B$, and then $u$ sends $B$ to this auxiliary node.
		By Lenzen's routing scheme, this is done in $O(1)$ rounds.

		Finally, all the comparators execute the comparisons, and if a swap is needed, the objects swap their metdata indices.
		Since a swap of wires in a comparator is equivalent to a swap of indices in our simulation, we have that after the last level, the metadata index of each object is its rank among the objects.

		To conclude, it takes $O(1)$ rounds to simulate each level of a sorting network.
		There are at most $O(\log (n^2))=O(\log n)$ levels to AKS sorting network, and therefore the running time for sorting objects of size $O(n)$ is $O(\log n)$ rounds.
	\end{proof}

	\subsection{Lower Bound}
	In this section we prove the lower bound of \cref{thm:sorting_esp_0}.

	\begin{lemma}\label{lem:LB}
		Every comparison based algorithm that solves \cref{prob_sorting} with $\eps=0$ must take $\Omega(\log n/\log\log n)$ rounds.
	\end{lemma}

	\begin{proof}
		Let $A$ be an algorithm that solves \cref{prob_sorting} with $n$ nodes and $n$ keys, each of size $\Theta(n)$, in $r$ rounds.
		We describe another algorithm $A'$ that also runs in $r$ rounds, solves \cref{prob_sorting} and performs $O(nr\log r)$ comparisons.
		Thus, for $r=o(\log n/\log\log n)$ we get a sorting algorithm of $n$ keys with $O(n\cdot r\log r)=o(n\log n)$ comparisons, which contradicts the celebrated comparison based sorting lower bound by Ford and Johnson~\cite{FJ59}.

		The algorithm $A'$ works as follows.
		Let us focus on one specific node $v$.
		$v$ will simulate $A$, but will ignore all the comparisons performed in $A$.
		Let $x_1,x_2,\dots$ be the keys $v$ receives during the running of the algorithm, due to their arrival order (breaking ties arbitrarily).
		In $A'$, the node $v$ maintains at any time the sorted order of all the keys that $v$ received so far.
		Whenever $v$ receives a new key $x_i$, $v$ performs a binary search on the keys $\{x_1,\dots,x_{i-1}\}$ (which are maintained in a sorted order).
		This takes $O(\log i)$ comparisons, and then $v$ updates the sorted order of all the keys $x_1,\dots,x_i$ with no additional comparisons.

		Since $v$ maintains at any time the sorted order of all the keys $v$ has received until this time, $v$ can simulate any operation that $v$ has to perform due to $A$, even if the operation requires some comparison between keys.

		We focus on the case where every key is of size $\Theta(n)$ words of space.
		In this case a node $v$ must get $\Omega(n)$ words to receive a key.
		Since $A'$ runs in $r$ rounds, $v$ gets at most $O(r)$ keys.
		The total number of comparisons $v$ performs while running $A'$ is $\sum_{i=1}^{O(r)}\log i=O(r\log r)$.
		There are $n$ nodes in the \cc, and therefore $O(n\cdot r\log r)$ comparisons in total across all the nodes.
	\end{proof}

\section{Answer Queries in the \cc Model}\label{sec:queries_routing}

		In this section we prove \cref{lem:queries_routing}.
		Recall that each node has $O(n)$ queries, such that each query is a pair of a resolving node, and the content of the query which is encoded in $O(\log n)$ bits.
		Moreover, the query can be resolved by the resolving node, and the size of the result is $O(\log n)$ bits.
		We show an algorithm that takes $O(1)$ rounds for each node to receive all the results of its queries.

	\begin{proof}
    [{Proof of \cref{lem:queries_routing}}]
		First, the algorithm sorts all the $O(n^2)$ queries due to the resolving nodes (i.e. the first element of the pair) with lenzen's sorting algorithm~\cite{Lenzen13} (in the sorted queries we assume that with each query we also have its ranking in the sorting).

		The next step goal is for all the nodes to know how many queries each node should resolve.
		To do so, the algorithm broadcasts $(r,q_r)$, where $r$ is the rank of query $q_r$, for every query $q_r$ such that the resolving node of $q_r$ is different than the resolving node of $q_{r+1}$ (notice that each node should send its first query in the sorting to the previous node for $q_r$ and $q_{r+1}$ which are not in the same node), as well as the $(r,q_r)$ of the last query in the sorting.
		By \cref{thm:gen_routing}, this is done in $O(1)$ rounds, since each node is the target of $O(n)$ messages, and each node locally computes the messages it desires to send from at most $O(n\log n)$ bits.
		Now, all the nodes know how many queries each node should resolve.

		For any resolving node $u_i$ let $s_i$ be the number of queries to node $u_i$.
		Notice that the average number of queries in the network is $\frac{1}{n}\sum_{i=1}^n s_i = O(n)$.

		\para{Creating auxiliary nodes.}
		For every original node $u_i$, the algorithm creates $a_i=\lceil{\frac{s_i}{n}}\rceil$ auxiliary nodes.
		Notice that  the total number of auxiliary nodes is at most $\sum_{i=1}^n a_i=
		\sum_{i=1}^n \lceil{\frac{s_i}{n}}\rceil \le
		\sum_{i=1}^n (1+\frac{s_i}{n})=
		n+\frac{1}{n}\sum_{i=1}^n s_i\le n+O(n)=O(n)$.
		By \cref{lem:aux_ndes} the algorithm can simulate each round with the auxiliary nodes on the original network in $O(1)$ rounds.

		Let $u_i$ be an original node with $a_i$ auxiliary nodes.
		Using \cref{thm:gen_routing}, after $O(1)$ rounds each of the $a_i$ nodes holds a copy of $u_i$'s information.

		Next, the $j$th query to node $u_i$ (which is easy to calculate since all the nodes know how many queries each node should resolve) is sent to the $\lceil{\frac{j}{n}}\rceil$ copy of $u_i$.
		Each node sends $O(n)$ queries and each copy of a resolving node receives at most $n$ queries.
		By Lenzen's routing scheme~\cite{Lenzen13} this is done in $O(1)$ rounds.

		Now, each of the $O(n^2)$ queries is in a copy its resolving node.
		The copies of the resolving nodes resolve the queries and send the results back to the original nodes, using  Lenzen's routing scheme~\cite{Lenzen13} (we assume that each query stores as a metadata the original node that holds the query).
	\end{proof}

\end{document}